\documentclass[11pt]{article}

\usepackage{amssymb}
\usepackage{amsmath}
\usepackage{graphicx}
\usepackage{color}
\usepackage{fullpage}

\usepackage{algorithmic}
\usepackage[linesnumbered,ruled]{algorithm2e}

\begin{document}

\newtheorem{theorem}{Theorem}[section]
\newtheorem{lemma}{Lemma}[section]
\newtheorem{corollary}{Corollary}[section]
\newtheorem{claim}{Claim}[section]
\newtheorem{proposition}{Proposition}[section]
\newtheorem{definition}{Definition}[section]
\newtheorem{fact}{Fact}[section]
\newtheorem{example}{Example}[section]

\newcommand{\quod}{\hfill $\blacksquare$ \bigbreak}
\newcommand{\reals}{I\!\!R}
\newcommand{\property}{I\!\!P}
\newcommand{\np}{\mbox{{\sc NP}}}
\newcommand{\sing}{\mbox{{\sc Sing}}}
\newcommand{\con}{\mbox{{\sc Con}}}
\newcommand{\prob}{\mbox{Prob}}
\newcommand{\atm}{\mbox{{\sc ATM}}}
\newcommand{\hopn}{\hop_{\cN}}
\newcommand{\atmn}{\atm_{\cN}}
\newcommand{\cA}{{\cal A}}
\newcommand{\cO}{{\cal O}}
\newcommand{\cP}{{\cal P}}
\newcommand{\cC}{{\cal C}}
\newcommand{\C}{{\cal C}}
\newcommand{\cB}{{\cal B}}
\newcommand{\cG}{{\cal G}}
\newcommand{\cN}{{\cal N}}
\newcommand{\cU}{{\cal U}}
\newcommand{\cF}{{\cal F}}
\newcommand{\cT}{{\cal T}}
\newcommand{\hx}{\hat{x}}
\newcommand{\cS}{{\cal S}}

\newcommand{\EE}{$\epsilon$-\-en\-ve\-lope}
\newcommand{\EQ}{$\epsilon$-\-con\-quer}
\newcommand{\ED}{{\tt En\-ve\-lo\-pe-Dis\-co\-ve\-ry}}
\newcommand{\SC}{{\tt Safe\--Con\-quer}}
\newcommand{\CE}{{\tt Con\-fir\-med-\-E\-cho}}
\newcommand{\CT}{{\tt Co\-lor\&\-Tran\-smit}}
\newcommand{\AC}{{\tt As\-sign-\-Co\-lor}}
\newcommand{\CD}{{\tt Con\-fli\-ct-\-De\-te\-ction}}
\newcommand{\CR}{{\tt Con\-fli\-ct-\-Re\-so\-lu\-tion}}
\newcommand{\SGR}{{\tt Shifted Grid-Refinement}}
\newcommand{\ERRD}{{\tt Error-Detection}}
\newcommand{\WGR}{{\tt Witnessed Grid-Refinement}}

\newcommand{\DB}{$\Delta$-block}
\newcommand{\DBs}{$\Delta$-blocks}
\newcommand{\FDB}{$5\Delta$-block}
\newcommand{\FDBs}{$5\Delta$-blocks}

\newcommand{\UB}{{\tt Universal Broadcast}}
\newcommand{\CAB}{{\tt Company-Aware Broadcast}}
\newcommand{\DI}{{\tt Dense-1}}
\newcommand{\DII}{{\tt Dense-2}}

\newcommand{\MS}{\mathcal{S}}

\newcommand{\eps}{{\epsilon}}
\newcommand{\la}{{\lambda}}
\newcommand{\al}{{\alpha}}
\newcommand{\qed}{\hfill $\square$ \smallbreak}

\newcommand{\UDGI}{{\tt UDG1}}
\newcommand{\UDGII}{{\tt UDG2}}
\newcommand{\SYM}{{\tt SYM}}

\newenvironment{proof}{\noindent{\bf Proof:}}{\qed}

\def\lalto{\left \lceil}
\def\ralto{\right \rceil}
\def\lbasso{\left \lfloor}
\def\rbasso{\right \rfloor}
\def\D{{\Delta}}
\def\qed{\hfill$\Box$}

\baselineskip    0.19in
\parskip         0.1in
\parindent       0.0in

\bibliographystyle{plain}
\title{{\bf Cost vs. Information Tradeoffs\\ for Treasure Hunt in the Plane }}
\author{
Andrzej Pelc \footnotemark[1] \footnotemark[2]
\and
Ram Narayan Yadav
}
\date{ }
\maketitle
\def\thefootnote{\fnsymbol{footnote}}

\footnotetext[1]{ \noindent
D\'epartement d'informatique, Universit\'e du Qu\'ebec en Outaouais, Gatineau,
Qu\'ebec J8X 3X7, Canada.
E-mails: {\tt pelc@uqo.ca}, {\tt narayanram1988@gmail.com} 
}
\footnotetext[2]{ \noindent  
Partially supported by NSERC discovery grant 2018-03899 and     
by the Research Chair in Distributed Computing at the
Universit\'e du Qu\'{e}bec en Outaouais. 
}

\begin{abstract}
A mobile agent has to find an inert treasure hidden in the plane. Both the agent and the treasure are modeled as points. This is a variant of the task known as treasure hunt.
The treasure is at a distance at most $D$ from the initial position of the agent, and the agent finds the treasure when it gets at distance $r$ from it, called the {\em vision radius}. However, the agent does not know the location of the treasure and does not know the parameters $D$ and $r$. The cost of finding the treasure is the length of the trajectory of the agent. We investigate the tradeoffs between the amount of information held {\em a priori}
by the agent and the cost of treasure hunt.
Following the well-established paradigm of {\em algorithms with advice}, this information is given to the agent in advance as a binary string, by an oracle cooperating with the agent and knowing 
the location of the treasure and the initial position of the agent. The size of advice given to the agent is the length of this binary string. 

For any size $z$ of advice and any $D$ and $r$, let $OPT(z,D,r)$ be the optimal cost of finding the treasure for parameters $z$, $D$ and $r$, if the agent has only an advice string of length $z$ as input.
We design treasure hunt algorithms working with advice of size $z$ at cost $O(OPT(z,D,r))$ whenever $r\leq 1$ or $r\geq 0.9D$. For intermediate values of $r$, i.e., $1<r<0.9D$, we design an almost optimal scheme of algorithms: for any constant $\alpha>0$, the treasure can be found at cost $O(OPT(z,D,r)^{1+\alpha})$.

\vspace*{1cm}

\noindent {\bf keywords:} mobile agent, treasure hunt, plane, advice.

\vspace*{2cm}

\end{abstract}

\thispagestyle{empty}

\pagebreak

\section{Introduction}

\subsection{The background and the problem}
Treasure hunt is the task of finding an inert target by a mobile agent in an unknown environment. In applications, the environment can be a communication network or a terrain,
and the agent can be a software agent looking for a piece of data in the first type of applications, or a mobile robot looking for an object, in the second case.
We consider treasure hunt in the plane, so the mobile agent may be thought of as a mobile robot or low-flying drone looking for a military target or for a lost person.
We assume that the target (treasure) is located at distance at most $D$ from the initial position of the agent, and that the agent finds the treasure when it gets at some distance $r>0$, called the {\em vision radius}. The agent does not know the parameters $D$ and $r$.  While the ignorance of $D$ is simply due to the fact that the treasure can be anywhere in the plane, the ignorance of $r$ can be due to physical conditions that affect the vision of the camera available to the robot in an unpredictable way: e.g., the density of the fog or of the grass in a savannah.

The mobile agent (robot) is modeled as a point moving along a polygonal line in the plane. It is equipped with a compass and a unit of length, and we assume that it has unbounded memory: from the computational point of view the agent is a Turing machine. Since the agent cannot learn anything during the execution of a treasure hunt algorithm until it sees the treasure, such an algorithm is simply a sequence of instructions of the type `` go at distance $x$ in direction $d$'',
without any conditional statements. 
The cost of a treasure hunt algorithm is the length of the trajectory of the agent from its initial position until it sees the treasure.

We investigate the tradeoffs between the amount of information held {\em a priori}
by the agent and the cost of treasure hunt.
Following the well-established paradigm of {\em algorithms with advice} (see the subsection ``Related work''), this information is given to the agent in advance as a binary string, by an oracle cooperating with the agent and knowing 
the location of the treasure and the initial position of the agent. The size of advice given to the agent is the length of this binary string. 
For a given size $z$ of advice, we want to find a treasure hunt algorithm of lowest possible cost, among algorithms using advice of this size.

Coming back to our application concerning finding a lost person, the size of advice of treasure hunt in the plane may be crucial. The lost person may have a GPS and hence may know their location in the plane. Also they know the position of the base. How to text little information to the rescuing team to allow a robot to reach the lost person fast? Time obviously matters, and both transmitting information and the travel of the robot take time, so it is important to know the tradeoffs. Also the transmitting device of the lost person may have little energy, so they may be able to transmit only a limited amount of information.

\subsection{Our results}

For any size $z$ of advice and any $D$ and $r$, let $OPT(z,D,r)$ be the optimal cost of finding the treasure for parameters $z$, $D$ and $r$, if the agent has only an advice string of length $z$ as input.
We design treasure hunt algorithms working with advice of size $z$ at cost $O(OPT(z,D,r))$ whenever $r\leq 1$ or $r\geq 0.9D$, i.e., in these ranges of vision radius our treasure hunt is optimal (up to multiplicative constants).
In the first range (for small vision radius), the cost of our algorithm is $O(D+\frac{D^2}{2^zr}(\log D + \log 1/r))$, while in the second range (for large vision radius) the cost of our algorithm is $O(D-r)$, hence it does not depend on the size of advice (and actually it does not require any advice in this case).
For intermediate values of $r$ (medium vision radius), i.e., $1<r<0.9D$, we design an almost optimal scheme of algorithms: for any constant $\alpha>0$, the treasure can be found at cost $O(OPT(z,D,r)^{1+\alpha})$. In this range, we design treasure hunt algorithms of cost $O((D+\frac{D^2}{2^zr})D^{\alpha})$, for any constant $\alpha>0$.
We then ``merge'' the above three solutions into one universal treasure hunt algorithm not knowing in which range of parameters it operates.

The above complexities show that, roughly speaking, the cost of treasure hunt in the plane decreases exponentially with growing size of advice, but this trade-off stops at the extremes: for small and medium vision range, cost $\Theta(D)$ cannot be beaten, even for very large advice, and for large vision radius, the cost is independent of advice and optimal treasure hunt does not use it. Our breaking points between the three ranges of vision radius are somewhat arbitrary, as  our algorithms
and lower bounds for small and medium vision radius have the same complexity for any constant $r$, and our algorithms and lower bounds for medium and large vision radius have the same complexity for any $r\in \Theta(D)$. The particular breaking points were chosen for convenience of analysis. We made no attempt at optimizing the multiplicative constants in the analysis.

The main difficulty of efficient treasure hunt in the plane, with little or no knowledge, lies in the need of reconciling two requirements: the agent should search
sufficiently far to reach the target at any, even very large distance, and at the same time, it should search sufficiently densely to find it even for small vision radius. If both $D$ and $r$ were known, this would not be hard to do because both the density of search and the limits of it would be fixed. With one of these parameters known, this  would not be hard either because the unknown parameter could be ``tried'' by doubling, and efficiency would follow from the telescopic effect. The real difficulty comes when neither $D$ nor $r$ are known because then some mechanism of systematic sweep of all possible hypotheses concerning couples $(D,r)$ is needed. This mechanism is easier in the case of small $r$ ($r\leq 1$), and was, in fact, designed in \cite{Pe} in the special scenario of no knowledge on the part of the agent. For arbitrary size of advice, in the case of small $r$, the difficulties are mostly geometric, but hypotheses concerning couples $(D,r)$ can still be sweeped  in the order of  growing diagonals of an infinite matrix, with growing possible values of $D$ in one dimension and decreasing possible values of $r$ in the other.
For large $r$ ($r\geq 0.9D$), this probing is also relatively simple because both parameters do not differ much and hence they can be tried simultaneously. In this range the algorithm is quite simple and the challenge is the geometric analysis of its complexity. However, the real difficulty comes in the intermediate range ($1<r<0.9D$). In this case, we will use a complex schedule of trying hypotheses concerning couples  $(D,r)$, in order to design our almost optimal scheme of treasure hunt algorithms.

\subsection{Related work}

{\bf Treasure hunt.}
The task of searching for a target (treasure) by mobile agents was investigated under various scenarios.
The environment where the target is hidden may be a graph or a plane, and the search may be deterministic or randomized.
The book \cite{AG} surveys both the search for a fixed target and the related rendezvous problem, where the target and the searching agent are both mobile and
they cooperate to meet. This book is concerned mostly with randomized search strategies. In \cite{MP,TSZ}, the authors studied relations between treasure hunt (searching for a fixed target) and rendezvous in graphs.  The authors of \cite{BCR} studied the task of finding a fixed point on the line and in the grid, and initiated the study of the task
of searching for an unknown line in the plane. This line of research was continued, e.g., in \cite{JL,La2}. In \cite{SF}, the authors concentrated on game-theoretic aspects of
the scenario where multiple selfish pursuers compete to find a target, e.g., in a ring. The main result of \cite{La} is an optimal algorithm to sweep a plane in order to locate an unknown fixed target, where locating means to get the agent originating at point $O$ to a point $P$ such that the target is in the segment $OP$. In \cite{FHGTM}, the authors considered the generalization of the search problem in the plane to the case of several searchers.  Efficient search for a fixed or a moving target in the plane, under complete ignorance on the part of the searching agent, was studied in \cite{Pe}. Hence the fixed target part of 
\cite{Pe} corresponds to our current problem for the special case of advice of size 0. However, while the results of \cite{Pe} are stated for any vision radius $r>0$, it was tacitly assumed that $r\leq 1$, and,
as explained in section 2, these results do not hold for arbitrary $r>0$.

{\bf Algorithms with advice.}
The paradigm of algorithms with advice was used predominantly for tasks in graphs.
Providing arbitrary items of information that can be used to increase efficiency of solutions to network problems 
 has been
proposed in \cite{AKM01,DP,EFKR,FGIP,FIP1,FIP2,FKL,FP,FPR,GPPR02,IKP,KKKP02,KKP05,MP,SN,TZ05}. This approach was referred to as
{\em algorithms with advice}.  
The advice, in the form of an arbitrary binary string, is given by a cooperating omniscient oracle either to the nodes of the network or to mobile agents performing some task in it.
In the first case, instead of advice, the term {\em informative labeling schemes} is sometimes used, if different nodes can get different information.

Several authors studied the minimum size of advice required to solve
network problems in an efficient way. 
In \cite{FIP1}, the authors compared the minimum size of advice required to
solve two information dissemination problems using a linear number of messages. 
In \cite{FKL}, it was shown that advice of constant size given to the nodes enables the distributed construction of a minimum
spanning tree in logarithmic time. 
In \cite{DKM,EFKR}, the advice paradigm was used for online problems.
In \cite{FGIP}, the authors established lower bounds on the size of advice 
needed to beat time $\Theta(\log^*n)$
for 3-coloring cycles and to achieve time $\Theta(\log^*n)$ for 3-coloring unoriented trees.  
In the case of \cite{SN}, the issue was not efficiency but feasibility: it
was shown that $\Theta(n\log n)$ is the minimum size of advice
required to perform monotone connected graph clearing.
In \cite{IKP}, the authors studied radio networks for
which it is possible to perform centralized broadcasting in constant time. They proved that constant time is achievable with
$O(n)$ bits of advice in such networks, while
$o(n)$ bits are not enough. In \cite{FPR}, the authors studied the problem of topology recognition in networks,  with advice given to the nodes. 
In \cite{DP}, the task of drawing an isomorphic map by an agent in a graph was considered, and the problem was to determine the minimum advice that has to be given to the agent for the task to be feasible. 
 Leader election with advice was studied in \cite{GMP} for trees, and in \cite{DiPe} for arbitrary graphs.
Graph exploration with advice was studied in \cite{BFU,GP} and treasure hunt with advice in graph environments was investigated in \cite{KKKS,MP}.
In a recent paper \cite{PY} we studied the size of advice sufficient to find a treasure in a geometric terrain with obstacles, at cost of optimal order of magnitude, where the vision radius of the agent is fixed to 1.

\section{Preliminaries}

We start with a correction concerning the results from \cite{Pe}, where treasure hunt in the plane was studied assuming that the agent has no {\em a priori} knowledge whatsoever, i.e., in our terms, treasure hunt with advice of size 0.
While the results of \cite{Pe} are stated for any vision radius $r>0$, it was tacitly assumed that $r\leq 1$, in which case they are valid. However,
as we will show below, these results do not hold for arbitrary $r>0$.

We first briefly recall the Algorithm  {\tt Static} from \cite{Pe}. 
The algorithm produces a trajectory of the mobile agent which is a polygonal line whose segments are parallel to the cardinal directions. 
For any positive real $x$, the instruction $(N,x)$ (resp. $(E,x)$, $(S,x)$ , and $(W,x)$) has the meaning ``go North (resp. East, South, and West) at distance $x$''. Juxtaposition is used for concatenation
of trajectories, and $\overline{T}$ denotes the trajectory reverse with respect to trajectory $T$. For any positive real $y$, let $Q(y)$ denote the square with side $y$ centered at the starting point of the
mobile agent. 

For any positve integers $k$ and $j$, the {\em spiral} $S(k,j)$ is the trajectory resulting from the following sequence of instructions:
$(E,2^{-j})$, $(S,2^{-j})$, $(W,2\cdot 2^{-j})$, $(N,2 \cdot 2^{-j})$, $(E,3\cdot 2^{-j})$, $(S,3 \cdot 2^{-j})$, $(W,4\cdot 2^{-j})$, $(N,4 \cdot 2^{-j})$, ..., $(E,(2k+1)\cdot 2^{-j})$, $(S,(2k+1) \cdot 2^{-j})$, $(W,(2k+2)\cdot 2^{-j})$, $(N,(2k+2) \cdot 2^{-j})$. Note that, during the traversal of the spiral $S(k,j)$, the mobile agent  gets at distance less than $2^{-j}$ from every point of the square $Q(2k\cdot 2^{-j})$.
Denote by  $\Pi(k,j)$ the trajectory $S(k,j)\overline{S(k,j)}$. 

Consider the infinite matrix $A$ whose rows are numbered by consecutive positive integers and whose columns are numbered by consecutive positive even integers. The term $A(i,j)$ in row $i$ and column $j$ is the trajectory $\Pi(2^{i+j},j)$. For any positive integer $i$, denote by $\Delta[i]$ the concatenation $\Pi(2^{i+2},2)\Pi(2^{i+3},4),\dots \Pi(2^{1+2i},2i)$ of trajectories in the $i$th diagonal of the matrix.
Algorithm {\tt Static} from \cite{Pe} is formulated as follows:
Follow the trajectory $\Delta[1]\Delta[2]\Delta[3]\dots$ until seeing the target.

Theorem 2.1 from \cite{Pe} states that 
the cost of Algorithm {\tt Static} is $O((\log D + \log \frac{1}{r}) D^2/r)$, where $D$ is an upper bound on the initial distance of the agent from the target and $r$ is the vision radius (called the sensing distance
in \cite{Pe}).
However, consider this algorithm for $D=2^a$ and $r=D/4$, where the treasure is at distance exactly $D$ from the initial position of the agent. Executing the algorithm, the agent sees the target while following the trajectory $\Delta[a]$. This means that it traversed the entire
trajectory $\Delta[a-1]$. This trajectory contains the trajectory $\Pi(2^{a+1},2)$ of length $2(2\cdot2^{a+1}+2)(2\cdot2^{a+1}+3)\cdot 2^{-2}\geq 2\cdot(2^{a+1})^2 = 8D^2$ and hence the cost of the algorithm is $\Omega(D^2)$ in this case. This contradicts the statement of Theorem 2.1 from  \cite{Pe} because for $D=2^a$ and $r=D/4$ we have $O((\log D + \log \frac{1}{r}) D^2/r)=O((\log \frac{D}{r})D^2/r)=O(D)$.

On the other hand, Theorem 2.2 from \cite{Pe} states that 
the cost of any treasure hunt algorithm with unknown bound $D$ on the initial distance and unknown vision radius $r$, is at least 
 $\frac{1}{16}((\log D + \log \frac{1}{r}) D^2/r)$, for some couple of parameters $D$ and $r$, for which this value is arbitrarily large.
 This in turn is refuted for couples of parameters $D$ and $r=D-\log D$ by Theorem \ref{large vision} from the present paper. Indeed, it follows from this theorem that treasure hunt (with no advice) can be accomplished for such parameters
 at cost $O(D-r)=O(\log D)$, while the  lower bound $\frac{1}{16}((\log D + \log \frac{1}{r}) D^2/r)$ is $\Omega(D)$ in this case.
 
 While for $r\leq1$ the results from \cite{Pe} are valid and our present results (presented in Section 4) generalize them for  arbitrary size of advice, the above comments show that the results from  \cite{Pe} do not remain valid for arbitrary $r>1$. The results of Sections 5 and 6 of the present paper, applied in the special case of $z=0$ (i.e., with no advice) can serve as a correction of \cite{Pe} for $r>1$: they give an almost optimal treasure hunt algorithm of cost $O((D+\frac{D^2}{r})D^{\alpha})$, for any fixed $\alpha>0$, in the range $1<r<0.9D$, and an optimal treasure hunt algorithm of cost $O(D-r)$ in the range $r \geq 0.9D$.
 
 Since for $r \geq D$, the agent can see the treasure from its initial position, without making any move, we assume throughout the paper that $r<D$.
 
 We will use the following terminology. The initial position of the agent is called $P$. The direction North-South is called {\em vertical} and the direction East-West is called {\em horizontal}. We will use the
 notion of {\em tiling}. This is a partition of the plane into squares of the same side length, called {\em tiles}, with all sides parallel or perpendicular to a given line $L$, and such that $P$ is a corner of one of the tiles. In order to make this a partition of the plane, we assume that each tile contains its two adjacent sides in a consistent way: in case of tiles with sides vertical and horizontal, these are the North and the East side of every tile. The side length of tiles is called the {\em size} of the tiling.
 
 For convenience, all our treasure hunt algorithms are formulated as infinite sequences of prescribed moves. It is understood that the algorithm is interrupted,
 i.e., the agent stops, as soon as it gets at the unknown distance $r$ from the treasure, at which time it sees the treasure. 
 
 \section{The advice and the basic traversal}

We first describe the advice given by the oracle that knows the initial position $P$ of the agent and the location $Q$ of the treasure, and that has $z$ available bits, where $z$ is a positive integer. (If $z=0$, no advice is given). The same advice will be used in all our algorithms. We call it the {\em canonical advice of size} $z$.
The oracle divides the plane into $2^z$ sectors, each with angle $2\pi/2^z$, using half-lines starting at $P$, one of which is in the direction North. Each sector consists of points between two consecutive half-lines $L$ and $L'$, where $L'$ is clockwise from $L$, including $L$ and excluding $L'$. The sector corresponding to lines $L$ and $L'$, where $L'$ forms the angle $i2\pi/2^z$ with direction North and $L$ forms the angle $(i+1)2\pi/2^z$ with direction North, for $0 \leq i  \leq 2^z-1$, is called the $i$-th sector. Angles are counted counterclockwise from the direction North.  
Let $0 \leq j  \leq 2^z-1$ be the number of the sector containing point $Q$. Let $w$ be the string of bits defined as the binary representation of $j$  padded by a prefix of $z-\lceil \log j \rceil$ zeroes. For example, if $z=4$ and $j=5$, the string $w$ is $(0101)$. The oracle gives the string $w$ to the agent.

Given the advice $w$, the agent decodes it as follows. It divides the plane into $2^z$ sectors, as described above, where $z$ is the length of $w$. Then it finds the integer $j$, whose binary representation is $w$. Finally it computes the $j$-th sector $S$ containing the location $Q$ of the treasure.

Suppose that  $z\geq 2$.
For given positive reals $D$ and $r$, such that $D > r$, and for a given sector  $S$ corresponding to half-lines $L$ and $L'$, where $L'$ is clockwise from $L$, 
we define the set of points $S^*$ which the intersection of the sector $S$ with the disc of radius $D$ centered at $P$.
The agent constructs a tiling of size $r$,
one of whose tiles has a corner at point $P$ with sides of tiles parallel or perpendicular to the line $L'$. 
Let $\Sigma$ be the set of tiles that intersect the set $S^*$. The set $\Sigma$ can be partitioned into columns of tiles, where a column is the set of tiles whose
centers lie on a line perpendicular to the line $L'$. Columns can be indexed by integers $1,2,\dots , t$, along line $L'$, starting from point $P$, see Fig. \ref{basic}.

\begin{figure}[tp]
\centering
\includegraphics[scale=0.85]{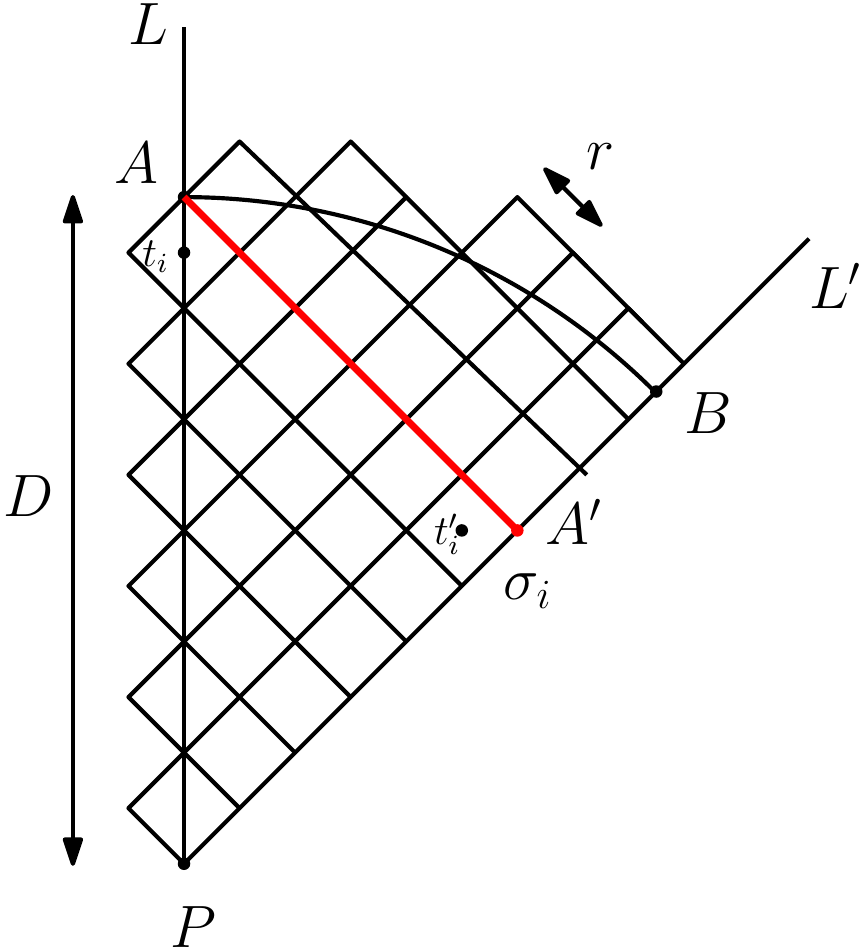}
\caption{The tiling of the 7-th sector $S$ with radius $D$ and $z=3$. The advice string $w$ is $(111)$. }
\label{basic}
\end{figure}

The aim of Procedure {\tt Basic Traversal with Advice} is to visit the centers of all tiles in the set $\Sigma$. This is done column by column, in increasing order of indices.
The procedure has input $w$, $D$ and $r$, where $w$ is a binary string, and $D > r$ are positive reals.  It can be described as follows.

\begin{center}
\fbox{
\begin{minipage}{11cm}

Procedure {\tt Basic Traversal with Advice} 

\vspace*{0.5cm}

Compute the sector $S$ using advice string $w$.\\
Compute the set $S^*$ using $S$ and $D$.\\
Compute the set $\Sigma$ of tiles using  $S^*$ and $r$.\\ 
Let $\sigma_1,\dots,\sigma_t$ be the columns of $\Sigma$.\\
Let $t'_i$ be the tile in $\sigma_i$ 
with one side in line $L'$.\\
Let $t_i$ be the tile in column $\sigma_i$ farthest from $t_i'$.\\
{\bf for} $i:=1$ {\bf to} $t$ {\bf do}\\
\hspace*{1cm}go to the center of tile $t_i'$\\
\hspace*{1cm}go to the center of tile $t_i$ on the line perpendicular to $L'$\\ 
 \hspace*{1cm}go back to the center of tile $t_i'$.

\end{minipage}
}
\end{center}

Suppose that the size of advice is $z\leq 1$. Then, for given positive reals $D$ and $r$, such that $D > r$, we define the spiral $X(D,r)$ (see Fig. \ref{spiral-large})
which is the trajectory resulting from the following sequence of instructions:
$(E,r)$, $(S,r)$, $(W,2 r)$, $(N,2  r)$, $(E,3 r)$, $(S,3  r)$, $(W,4 r)$, $(N,4  r)$, ..., $(E,(2k+1) r)$, where $k=\lceil D/r\rceil$. Note that, if the vision radius is $r$ and the treasure is located at distance at most $D$ from the initial position of the agent, then during the traversal of the spiral $X(D,r)$, the agent  must see the treasure. 

Now we are able to formulate Algorithm {\tt Basic Traversal}. This is a treasure hunt algorithm working under the assumption that an upper bound $D$ on the distance of the treasure from the initial position of the agent, and the vision radius $r$ are known to the agent. This algorithm will be used as a building block in our treasure hunt algorithms ignoring parameters $D$ and $r$.

\begin{center}
\fbox{
\begin{minipage}{9cm}

{\bf Algorithm} {\tt Basic Traversal}\\

\noindent
{\bf if} the size $z$ of advice is at least 2 {\bf then}\\
 \hspace*{1cm}call Procedure {\tt Basic Traversal with Advice}\\
{\bf else}\\
 \hspace*{1cm}follow trajectory $X(D,r)$

\end{minipage}
}
\end{center}

The following lemma estimates the number of tiles intersecting the set $S^*$. It will be used to estimate the cost of Algorithm {\tt Basic Traversal}.

\begin{lemma}\label{max-tile}
Let $D$ and $r$ be positive reals such that 
$r<D$, and let $z$ be an integer larger than 1. Let ${S}$ be any sector corresponding to half-lines $L$ and $L'$ forming an angle $2\pi/2^z$, where $L'$ is clockwise from $L$. Let $P$ be the intersection point of lines $L$ and $L'$. Let $S^*$ be the set of points which is the intersection of the sector $S$ with the disc of radius $D$ centered at $P$. Then the number of tiles of size $r$ whose sides are parallel or perpendicular to the line $L'$, and that intersect the set $S^*$, is at most $69(\frac{D^2}{2^z.r^2}+\frac{D}{r})$.
\end{lemma}

\begin{proof}
Let $k = 2^z$. The area of the set $S^*$ is $\frac{\pi D^2}{k}$. The number of tiles of size $r$ with sides parallel or perpendicular to the line $L'$ that are  contained in the set $S^*$ is $N_1 \leq \frac{\pi D^2}{kr^2}$. We estimate the number $N_2$ of tiles of size $r$ that intersect the perimeter of the set $S^*$. The length of the perimeter of $S^*$ is $(2D+2\pi D/k)$. Any segment of length at most $r/4$ of this perimeter intersects at most 4 tiles of size $r$. Hence, $N_2 \leq \lceil\frac{(2D+2\pi D/k)}{r/4}\rceil \cdot 4 \leq \frac{32D}{r}+\frac{32\pi D}{kr} + 4$. We first conclude the proof under the additional assumption that $\frac{D}{r} \geq k$. \\
Then, since $k\geq 4$, we have $N_2 \leq \frac{33D}{r}+\frac{32\pi D}{kr} \leq \frac{65D}{r}$. Therefore, the total number of tiles intersecting $S^*$ is at most $N \leq N_1+N_2 \leq \frac{\pi D^2}{kr^2} + \frac{65D}{r} \leq \frac{\pi D^2}{kr^2} + \frac{65D^2}{kr^2} \leq (\pi+65)\frac{D^2}{kr^2}$. 

We now remove the additional assumption. 
Suppose that $\frac{D}{r} < k$. Let $A$ be the point on line $L$ at distance $D$ from $P$.
Let $A'$ be the point on the line $L'$ such that $AA'$ is perpendicular to the line $L'$, see Fig. \ref{basic}.

Let $y=|AA'|=D \sin (2 \pi/k)$. We have $\lim_{k \to \infty} \frac{\sin (2\pi/k)}{2\pi/k}=1$, and hence $y= 2\pi D/k \leq 3 \pi D/k$, for sufficiently large $k$. In this case, the set $S^*$ is contained in a rectangle with one side of length $D$ and another side of length $y$.
Since sides of tiles are parallel or perpedicular to the line $L'$, the number of tiles intersecting this rectangle is at most $\lceil y/r \rceil \cdot \lceil D/r \rceil$. 
We have $\lceil y/r \rceil  \leq \lceil  (3\pi D)/(kr) \rceil \leq 12 D/(kr) +1$. Since $D<rk$, we have  $12 D/(kr) +1 \leq 13$. Hence $\lceil y/r \rceil \leq 13$. This implies $\lceil y/r \rceil \cdot \lceil D/r \rceil \leq 13 \cdot \lceil\frac{D}{r}\rceil \leq 13 \cdot (\frac{D}{r}+1)$. Since $r<D$, we have $ 13 \cdot (\frac{D}{r}+1) \leq 26\frac{D}{r}$. 

Hence, the number of tiles of size $r$ intersecting  $S^*$ is always at most $((\pi+65)\frac{D^2}{2^zr^2}+26\frac{D}{r}) \leq 69(\frac{D^2}{2^zr^2}+\frac{D}{r})$.
This proves the lemma.
\end{proof}

\begin{figure}[tp]
\centering
\includegraphics[scale=1.5]{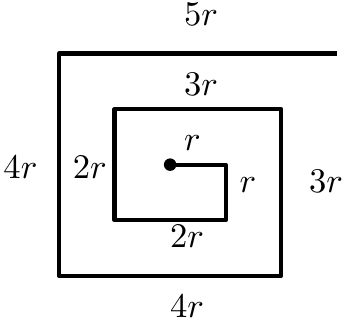}
\caption{The spiral $X(D,r)$ for $k=\lceil D/r\rceil=2$}
\label{spiral-large}
\end{figure}

The next lemma proves the correctness of Algorithm {\tt Basic Traversal} and estimates its cost.

\begin{lemma}\label{cost-basic}
Fix reals $D,r$ where $D$ is an upper bound on the distance between the initial position of the agent and the treasure, and $r$ is the radius vision of the agent,  such that $r < D$.   For advice of size $z\geq0$, Algorithm {\tt Basic Traversal} correctly finds the treasure and works at cost at most $138\cdot(\frac{D^2}{2^zr}+D)$.
\end{lemma}
\begin{proof}
We consider the following two cases.  

\noindent
Case 1. Advice of size $z \geq 2$ \\
Using the canonical advice of size $z$, the agent computes the sector $S$ with angle of size $2\pi/2^z$ which contains the treasure. Let $S^*$ be the set of points which is the intersection of the sector $S$ with the disc of radius $D$ centered at $P$. 
In this case the agent executes the Procedure {\tt Basic Traversal with Advice} and visits the center of each tile intersecting the set $S^*$. Since the size of the tiling is $r$, the agent sees all the points of a  tile from its center. Let $T$ be the trajectory of the agent produced by the execution of the Procedure {\tt Basic Traversal with Advice}. Since every point in the set $S^*$ is at distance at most $r$ from some point of $T$, the agent must see the treasure by the end of the execution of Procedure {\tt Basic Traversal with Advice}. This proves the correctness of the algorithm in this case.

Next, we estimate the cost of our algorithm in this case. By Lemma \ref{max-tile}, the number of tiles intersecting the set $S^*$ is at most $69\cdot(\frac{D^2}{2^z.r^2}+\frac{D}{r})$. In the execution of Procedure {\tt Basic Traversal with Advice}, the agent makes the first move from $P$ to the center of the tile containing it, and all other moves from the center of a tile to the center of an adjacent tile. The center of each tile intersecting $S^*$ is visited at most twice, and each such move is at distance at most $r$. Hence the cost of the algorithm is at most $2 \cdot 69\cdot(\frac{D^2}{2^zr^2}+\frac{D}{r})r =138\cdot(\frac{D^2}{2^zr}+D)$ in this case.

\noindent
Case 2. Advice of size $z < 2$ \\
In this case the agent follows the trajectory $X(D,r)$. During the traversal of the trajectory the agent gets at distance at most $r$ from every point of the square of side $2k r$ centered at the initial position $P$ of the agent, where $k=\lceil D/r \rceil$. Since the treasure is located at distance at most $D$ from $P$, and the vision radius is $r$, the agent must see the treasure by the end of the execution of the algorithm, which proves correctness in this case. 

Next, we estimate the cost of our algorithm  in this case as follows. The length of the trajectory $X(D,r)$ is $2r[1+2+\dots +2k]+(2k+1)r \leq 2r[1+2+\dots +(2k+1)] \leq r[(2k+1)(2k+2)] \leq  r(2k+2)^2$. Since $k=\lceil D/r \rceil$, we have $r(2k+2)^2=4r(\lceil D/r \rceil +1)^2 \leq 4r(D/r +2)^2$. Since $r < D$, we have $4r(D/r +2)^2 \leq 4r (3D/r)^2$. Hence, the length of the trajectory of the agent is at most $36D^2/r$, which is at most $72\cdot\frac{D^2}{2^zr}$
because $z\leq 1$. Hence the cost of the algorithm is at most $138\cdot(\frac{D^2}{2^zr}+D)$ in this case as well.
\end{proof}

The following lemma establishes a lower bound on the cost of any treasure hunt algorithm using advice of size $z$.

\begin{lemma}\label{lb-medium}
Suppose that the treasure is at distance at most $D$ from the initial position of the agent, and that the vision radius is
$r<0.9D$. Then the cost of any treasure hunt algorithm using advice of size $z \geq 0$ is at least $\frac{1}{800}(\frac{D^2}{2^z r} + D)$. 
\end{lemma}

\begin{proof}
The obvious lower bound on the cost of  treasure hunt is $D-r$, as the treasure can be at distance exactly $D$ from the initial position of the agent. Hence, in the case when $r<0.9D$, the cost of treasure hunt is at least $D-r \geq D/10$. Next, we consider the cost of treasure hunt using any advice of size $z \geq 0$.
Consider the square $S$ of side $\sqrt{2} D/2$ with sides vertical and horizontal and with the South-West corner at the starting position $P$ of the agent, see Fig. \ref{lb-medium-fig}. 

\begin{figure}[tp]
\centering
\includegraphics[scale=0.8]{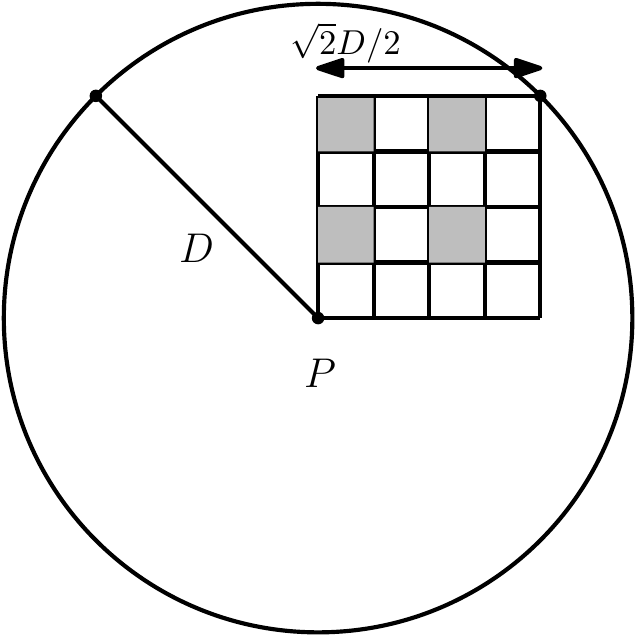}
\caption{Illustration of the proof of Lemma \ref{lb-medium} }
\label{lb-medium-fig}
\end{figure}

Consider the tiling of size $2r$ with sides vertical or horizontal. 
The number $N$ of tiles included in the square $S$ is at least $(\lfloor \frac{\sqrt{2} D}{4r} \rfloor)^2 \geq (\frac{\sqrt{2} D}{4r}-1)^2$. 
Rows of tiles included in $S$ are indexed $1,2, \dots$ from the North side of $S$ going South and columns of tiles included in $S$ are indexed $1,2,\dots$ from the West side of $S$ going East. 

Case 1. $D/r >5$\\ 
In this case we have $r < D/5$. Hence $N \geq (\frac{\sqrt{2} D}{4r}-1)^2 \geq (\frac{\sqrt{2} D}{4r}-\frac{D}{5r})^2 \geq (\frac{D}{10r})^2$. Consider the center of every other tile of odd-indexed tile rows in $S$ as possible locations of the treasure, see Fig. \ref{lb-medium-fig} (we call them shaded tiles in the rest of the proof). Hence, the number of such possible locations is at least $\frac{1}{4}(\frac{D}{10r})^2$. Using advice of size $z$, we have $2^z$ different advice strings. By the Pigeonhole Principle, there are at least $N^*=\frac{1}{400}(\frac{D^2}{2^z r^2})$ shaded tiles (with possible locations of the treasure at their center) corresponding to the same advice string. In order to see the treasure at the center of a tile, the agent has to be inside this tile. Hence,  
if the treasure is  at the center of one of these $N^*$ tiles, the agent must visit each of these tiles. 
Since the distance between any two shaded tiles is at least $r$, the cost of any treasure hunt algorithm using advice of size $z$ must be at least $\frac{1}{400}(\frac{D^2}{2^z r})$. As mentioned before, the cost of any treasure hunt algorithm is at least $D/10$. Hence we conclude that
the cost of any treasure hunt algorithm using advice of size $z$ must be at least $\frac{1}{800}(\frac{D^2}{2^z r} + D)$.\\

Case 2. $D/r \leq 5$.\\
As mentioned before, we have the obvious lower bound $D-r \geq D/10$. Since $D/r \leq 5$ and $z \geq 0$, we have $\frac{1}{800}(\frac{D^2}{2^z r}+D) \leq \frac{1}{800}(\frac{5D}{2^z}+D) \leq \frac{1}{800}\cdot 6D \leq D/10$.

Hence, in all cases, we get the lower bound $ \frac{1}{800}(\frac{D^2}{2^z r} + D)$ on the cost of any treasure hunt algorithm using advice of size $z$.     
\end{proof}

\section{Small vision radius}
In this section, we consider the case of small vision radius i.e. $r \leq 1$. Our aim is to design an algorithm working at cost
$O(D+\frac{D^2}{2^zr}(\log D + \log 1/r))$ (which, as we will later show, is optimal) where the size of advice is $z$ and the distance between the location of the treasure and the initial position of the agent is at most $D$. The high-level idea of the algorithm is to make consecutive hypotheses $D'$ and $r'$, concerning $D$ and $r$ respectively, where $D'=2^i$
and  $r'=1/2^j$, and for each of those hypotheses execute Algorithm {\tt Basic Traversal} and backtrack to the initial position, until the treasure is seen.
This idea follows  that from \cite{Pe}. However, similarly to the cost $O(\frac{D^2}{r}(\log D + \log 1/r))$ obtained in \cite{Pe}, the straightforward application of this idea  would result in the cost
$O((D+\frac{D^2}{2^zr})(\log D + \log 1/r))$, i.e., the summand $D$ would also be multiplied by $\log D + \log 1/r$. This did not hurt in \cite{Pe} because, in the absence of advice, i.e., for $z=0$, the term $\frac{D^2}{r}$ always dominates $D$ in view of $r<D$. In our case, $D$ may dominate
$\frac{D^2}{2^zr}(\log D + \log 1/r)$ for large enough $z$, and then $\Theta((D+\frac{D^2}{2^zr})(\log D + \log 1/r))=\Theta((D(\log D + \log 1/r))$
but $\Theta(D+\frac{D^2}{2^zr}(\log D + \log 1/r))=\Theta(D)$. Thus our cost would not be optimal in this case. In order to take care of this possibility, we need two separate algorithms: one working 
at cost  $O(\frac{D^2}{2^zr}(\log D + \log 1/r))$, and the other working at cost $O(D)$, depending on which term dominates. We merge the two algorithms
by interleaving trips following trajectories determined by each of them at exponentially growing distances and each time backtracking to the starting position.
If $z$ is very small ($z\leq 1$), the term $\frac{D^2}{2^zr}(\log D + \log 1/r)$ dominates $D$, and hence we can apply the sequential probing of hypotheses concerning $D$ and $r$ in a straightforward way.

We now proceed to the detailed description of the algorithm.
Let $T(i,j)$ be the trajectory of the agent resulting from the execution of Algorithm {\tt Basic Traversal} with advice of size $z$ for input $D=2^i$ and $r=1/2^j$. 

Consider the infinite matrix $B$ whose rows are numbered by consecutive positive integers and columns are numbered by consecutive positive even integers.. The term $B(i,j)$ in row $i$ and column $j$ is the trajectory $T(i,j)\overline{T(i,j)}$,
where $\overline{T}$ denotes the trajectory reverse with respect to $T$.

For any positive integer $i$, denote by $\Gamma[i]$ the concatenation of trajectories $B(i,2)$, $B(i-1,4)$,...$B(1,2i)$ in the $i$th diagonal of the matrix $B$. Let $\Pi_1$ be the infinite trajectory resulting from the concatenation of trajectories $\Gamma[1],\Gamma[2],\dots$. 

If $z\leq 1$, we simply follow the trajectory $\Pi_1$.
Suppose that $z \geq 2$. Using the advice string, the agent finds the sector $S$ with angle $2\pi/2^z$ which contains the treasure. Let $L$ and $L'$ be the half-lines corresponding to the sector $S$, where $L'$ is clockwise from $L$,
and let $\Pi_2$ be the infinite trajectory following the line $L'$  from the initial position of the agent.  

We formulate our algorithm as follows. It is interrupted when the agent sees the treasure.

\begin{center}
\fbox{
\begin{minipage}{9cm}

{\bf Algorithm} {\tt Small vision}\\

\noindent
{\bf if} the size $z$ of advice is at least 2 {\bf then}\\
\hspace*{0.5cm}$p:=1$\\
\hspace*{0.5cm}{\bf repeat}\\
 \hspace*{1cm}Go at distance $2^p$ along the trajectory $\Pi_1$\\
 \hspace*{1cm}Backtrack to the initial position\\
 \hspace*{1cm}Go at distance $2^p$ along the trajectory $\Pi_2$\\
 \hspace*{1cm}Backtrack to initial position\\
 \hspace*{1cm}$p:=p+1$\\
{\bf else}\\
 \hspace*{1cm} Follow the trajectory $\Pi_1$
\end{minipage}
}
\end{center}

Before proving the correctness and estimating the cost of our algorithm we state the following simple geometric observation, see Fig. \ref{long-2}.

\begin{figure}[tp]
\centering
\includegraphics[scale=1.0]{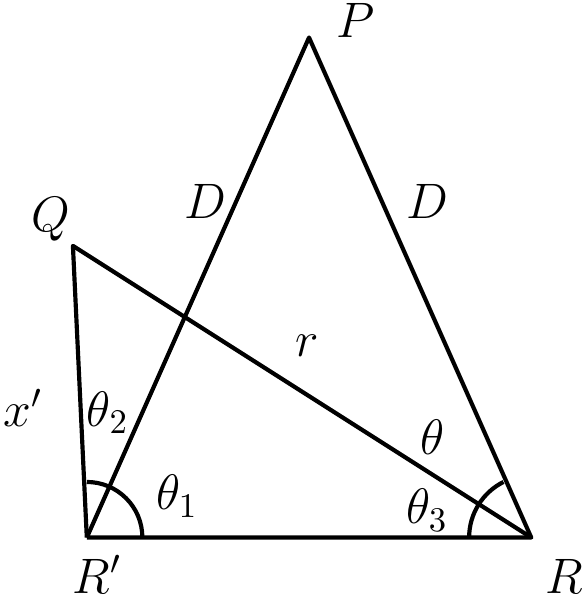}
\caption{Illustration for the proof of Lemma \ref{geom}}
\label{long-2}
\end{figure}

\begin{lemma}\label{geom}
Consider an isosceles triangle $PRR'$ such that $|PR|=|PR'|$. Let $Q$ be a point such that $\angle QR'R > \angle PR'R$ and $\angle QRR' <\angle PRR'$.
Then $|QR'|<|QR|$.
\end{lemma}

\begin{proof}
Since the triangle $PR'R$ is isosceles, we have $\angle PRR'= \angle PR'R= \theta_1$, see Fig. \ref{long-2}. Consider the triangle $QR'R$. Let $\angle QR'R =\theta_1+ \theta_2$ and $\angle QRR'= \theta_3$. By the definition of the points $Q$ and $P$, we have $\theta_3 < \theta+\theta_3=\theta_1$ and $\theta_1 < \theta_1+\theta_2$, see Fig. \ref{long-2}. Hence $\theta_3<\theta_1+\theta_2$ and thus $|QR'|< |QR|$ because, in the 
triangle $QR'R$, a larger angle must face a larger side. 
\end{proof}

The following theorem proves the correctness and estimates the cost of Algorithm {\tt Small vision}.

\begin{theorem}\label{small-vision-cost}
Suppose that the treasure is at distance at most $D$ from the initial position of the agent and the vision radius $r$  is at most 1, where parameters $D$ and $r$ are unknown to the agent. Then an agent executing Algorithm {\tt Small vision} finds the treasure at cost $O(D+\frac{D^2}{2^zr}(\log D + \log 1/r))$.
\end{theorem}

\begin{proof}
We first prove the correctness of the algorithm. 
Let $a=\lceil \log D \rceil$ and let $b$ be the smallest even integer greater or equal to $\lceil \log \frac{1}{r}\rceil$. Thus $b \leq \lceil \log \frac{1}{r}\rceil +1$. 
The agent sees the treasure by the time when it traverses the trajectory $B(a,b)$ which is a sub-trajectory of $\Pi_1$. Since in both cases, the agent follows $\Pi_1$ arbitrarily far until it sees the treasure, at some point it will traverse the trajectory $B(a,b)$, which proves correctness.

We now estimate the cost of the algorithm. We consider the following two cases.\\

\noindent

Case 1. Advice of size $z < 2$ \\
Let $a = \lceil \log D \rceil$ and let $b$ be the smallest even integer greater or equal to $\lceil \log \frac{1}{r}\rceil$.
When $z<2$, the trajectory $B(i,j)$ follows the spiral $X(2^i,2^{-j})$ and backtracks on it. In this case, the analysis follows closely that from the proof of Theorem 2.1 in \cite{Pe}. 
The length of the trajectory $B(i,j)$ is  $4\cdot 2^{-j}[1+2+\dots +2k]+2(2k+1)2^{-j} \leq 4\cdot2^{-j}[1+2+\dots +(2k+1)] \leq 2\cdot 2^{-j}[(2k+1)(2k+2)] \leq 2\cdot 2^{-j}(2k+2)^2$, where $k=\lceil 2^i/2^{-j} \rceil> 1$. Hence,  the length of trajectory $B(i,j)$ is at most $32k^2 2^{-j} \leq 32 \cdot 2^{2i+2}\frac{1}{2^{-j}} = 128\cdot 2^{2i+j}$. Hence the length of the trajectory $\Gamma[i]$ is at most $128[2^{(2i+2)}+2^{(2(i-1)+4)}+2^{(2(i-2)+6)}\dots+2^{(2+2i)}]=128 i \cdot 2^{(2i+2)}$. The term $B(a,b)$ of the matrix $B$ is in the $(a+\frac{b}{2}-1)$th diagonal. The cost of Algorithm {\tt Small vision} in this case is at most the sum of lengths of trajectories $\Gamma[1], \Gamma[2], \dots, \Gamma[a+\frac{b}{2}-1]$, which is at most $2\cdot 128 i \cdot 2^{(2i+2)}$, where $i = a+\frac{b}{2}-1$. Since $z < 2$ in this case, the cost is at most $2\cdot 128 i \cdot 2^{(2i+2)} \leq 2048 \cdot \frac{1}{2^z}\cdot i 2^{2i}$. Since $a=\lceil \log D \rceil$ and $\frac{b}{2}-1 \leq \frac{1}{2} \lceil \log \frac{1}{r}\rceil$,
we have $i \in O(\log D + \log 1/r))$. By definition, $2^a \leq 2D$ and $2^b\leq 4/r$. Hence $2^{2i} \leq 2^{2a+b} \leq 16\frac{D^2}{r}$.
Hence the cost is in $O(\frac{D^2}{2^z r}(\log D+\log \frac{1}{r}))$, and thus in $O(D+\frac{D^2}{2^zr}(\log D + \log 1/r))$.\\

\noindent

Case 2.  Advice of size $z \geq 2$ \\
We consider the following two subcases. \\

Subcase 2.1. $\frac{D}{r} \leq \frac{1}{\sin (2\pi/2^z)}$\\

Consider  phase $p=\lceil \log D \rceil$ in the {\bf repeat} loop of Algorithm {\tt Small vision}.

{\bf Claim.} 
The agent sees the treasure following the trajectory $\Pi_2$ by the end of the phase $p$ of the Algorithm {\tt Small vision}. 

In order to prove the claim, consider the canonical advice of size $z$. Using the advice string, the agent finds the sector $S$ with angle $2\pi/2^z$ which contains the treasure. Let $L$ and $L'$ be the half-lines corresponding to the sector $S$, where $L'$ is clockwise from $L$. Let $S^*$ be the set of points which is the intersection of the sector $S$ with the disc of radius $D$ centered at $P$. Let $R$ (respectively $T$) be the point on the half-line $L$ (respectively $L'$) such that the distance between the initial position $P$ of the agent and the points $R$ and $T$ is $D$. Let $Q$ be the point on the half line $L'$ such that the line $RQ$ is perpendicular to $L'$ and hence $|PQ|\leq D $, see Fig. \ref{claim}. Let $|RQ|=x$. Using the triangle $PQR$, we have $\sin (2\pi/2^z) = x/D$ i.e $\frac{D}{x}=\frac{1}{\sin (2\pi/2^z)}$. In this case, since $\frac{D}{r} \leq \frac{1}{\sin (2\pi/2^z)}$, we have $\frac{D}{r} \leq \frac{D}{x}$ and hence $x \leq r$. Let $S'$ be the set of points inside the triangle $PQR$. Since in the phase $p$ of Algorithm {\tt Small vision} the agent goes on the half-line $L'$ at distance $2^p \geq D$, and in view of $x \leq r$, the agent must see all the points in the set $S'$ while following the trajectory $\Pi_2$, by the end of the phase $p$.

Next, we show that the agent also sees all the points inside the set $(S^*\setminus S')$ from the point $Q$.

\begin{figure}[tp]
\centering
\includegraphics[scale=0.9]{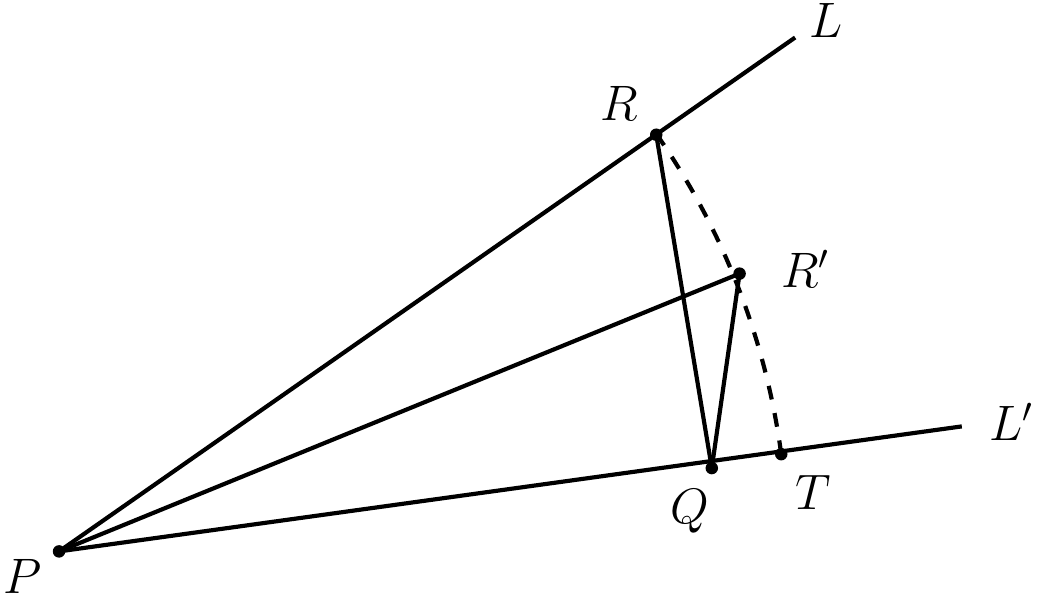}
\caption{Illustration of the proof of Theorem \ref{small-vision-cost} }
\label{claim}
\end{figure}

Let $R'$ be any point on the arc $RT$. Since $|PR|=|PR'|=D$, the triangle $PR'R$ is isosceles (see Fig. \ref{claim}).
By Lemma \ref{geom}, we have $|QR'|<|QR|< r$.
Hence, from the point $Q$, the agent can see all the points in the arc $RT$, and hence also all the points in the set $S^*\setminus {S}'$. 

It follows that when the agent reaches point $Q$ following the trajectory $\Pi_2$, it sees every point in $S^*$, and thus it must see the treasure by this time. This proofs the claim.

In view of the claim, the length of the trajectory of the agent until it finds the treasure is at most the sum of the lengths of trajectories traversed by the agent until the end of phase $p$ of Algorithm {\tt Small vision}. In each phase $q=1,2,\dots$ of Algorithm {\tt Small vision}, the agent first goes at  distance $2^q$ along the trajectory $\Pi_1$ and backtracks, and then goes at  distance $2^q$ along the trajectory $\Pi_2$ and backtracks. Hence, the cost of the algorithm in this case is at most $4(2^1+2^2+2^3+,\dots,+2^{\lceil \log D \rceil}) \leq 4\cdot 2\cdot 2^{\lceil \log D \rceil} \leq 16D$. Hence the cost of the algorithm is in
 $O(D+\frac{D^2}{2^zr}(\log D + \log 1/r))$ in this subcase.\\

Subcase 2.2. $\frac{D}{r} > \frac{1}{\sin (2\pi/2^z)}$\\

In this case, we use the fact that the agent always sees the treasure while following the trajectory $\Pi_1$, until it completely traverses the sub-trajectory
$B(a,b)$, where $a = \lceil \log D \rceil$ and $b$ is the smallest even integer greater or equal to $\lceil \log \frac{1}{r}\rceil$. By Lemma \ref{cost-basic}, the length of $B(i,j)$ is at most $2\cdot 138(\frac{D'^2}{2^z r'}+D')$, where $D'=2^i$ and $r'=2^{-j}$. For any $i$, the length of  the trajectory $\Gamma[i]$ is the sum of the lengths of the trajectories $B(i,2), B(i-1,4),\dots,B(1,2i)$, which is at most $2\cdot \frac{138}{2^z}\{[2^{(2i+2)}+2^{(2(i-1)+4)}+2^{(2(i-2)+6)}\dots+2^{(2+2i)}] + [2^i+2^{i-1}+,\dots,+2^2+2^1]\} \leq 2\cdot \frac{138}{2^z}[i\cdot 2^{2i+2} + 2\cdot 2^i]$. The agent sees the treasure while following the trajectory $\Pi_1$ in the smallest phase $p$ such that $2^p$ is at least the sum of lengths of trajectories $\Gamma[1],\Gamma[2],\dots,\Gamma[a+\frac{b}{2}-1]$. Since the agent also follows the trajectory $\Pi_2$ in each phase and backtracks every time, the cost of the algorithm in this case is at most $16[2\cdot 2\cdot \frac{138}{2^z}(i\cdot 2^{2i+2}+2\cdot2^i)] \leq 64 \cdot \frac{138}{2^z} \cdot 4[i \cdot 2^{2i}+2^i]$, where $i=a+\frac{b}{2}-1$. By definition, $2^a \leq 2D$ and $2^b\leq 4/r$. Hence $2^{2i} \leq 2^{2a+b} \leq 16\frac{D^2}{r}$. Since $2^i \geq 1$, the cost of the algorithm is at most  $64 \cdot \frac{138}{2^z} \cdot 4[i \cdot 2^{2i}+2^i] \leq 64 \cdot \frac{138}{2^z} \cdot 4[(i+1) \cdot 2^{2i}] \leq 64 \cdot \frac{138}{2^z} \cdot 4[(i+1) \cdot 16D^2/r]$. Since $a=\lceil \log D \rceil$ and $\frac{b}{2} \leq \frac{1}{2} \lceil \log \frac{1}{r}\rceil + 1$, we have $(i+1) \in O(\log D + \log 1/r))$. Hence the cost is in $O(\frac{D^2}{2^z r}(\log D+\log \frac{1}{r}))$, and thus in $O(D+\frac{D^2}{2^zr}(\log D + \log 1/r))$.

Hence in all the cases, we get the upper bound $O(D+\frac{D^2}{2^zr}(\log D + \log 1/r))$ on the cost of our algorithm using advice of size $z$.
\end{proof}

The following lower bound shows that Algorithm {\tt Small vision} has cost optimal among treasure hunt algorithms using advice of size $z$.

\begin{theorem}\label{small-lb}
The cost of any treasure hunt algorithm using advice of size $z$, with unknown
bound $D$ on the distance between the starting position of the agent and the initial location of the treasure and unknown vision radius $r$ which is at most 1, is $\Omega(D+\frac{D^2}{2^zr}(\log D + \log 1/r))$. 
\end{theorem}

\begin{proof}
The obvious lower bound on the cost of any treasure hunt algorithm is $D-r$. Since $r \leq 1$, the cost of any treasure hunt algorithm must be at least $D-1$. Hence, the cost is in $\Omega(D)$. Thus it is enough to prove the lower bound $\Omega(\frac{D^2}{2^zr}(\log D + \log 1/r))$. We prove the following claim.

{\bf Claim.} Consider a square $S$ of side $x$ such that $x = 2kr$, where $k$ is a positive integer, with sides horizontal or vertical. If the treasure is hidden in the square $S$ and vision radius is $r$, then the cost of any treasure hunt algorithm using advice of size $z$ is at least $\frac{1}{16}\cdot \frac{x^2}{2^z r}$.   \\
In order to prove the claim,
partition the square $S$ into square tiles of size $2r$ with sides horizontal or vertical. Tile rows are indexed $1,2, \dots$ from the North side of $S$ going South and tile columns are indexed $1,2,\dots$ from the West side of $S$ going East. The number of tiles included in the square $S$ is $\frac{x^2}{4 r^2}$. Consider the center of every other tile of odd-indexed tile rows in $S$ as possible locations of the treasure (we call them special tiles in the rest of the proof). The number of such possible locations is at least $\frac{1}{4}\cdot \frac{x^2}{4 r^2}=\frac{1}{16}\cdot \frac{x^2}{r^2}$. Using advice of size $z$, we have $2^z$ different advice strings. By the Pigeonhole Principle, there are at least $N= \frac{1}{16}\cdot \frac{x^2}{2^z r^2}$ special tiles corresponding to the same advice string. Since the side of a tile is $2r$, in order to see the treasure at the center of a tile, the agent has to be inside this tile. Hence, if the treasure is at the center of one of the $N$ special tiles, the agent must visit each of these tiles. Since the distance between any two special tiles is at least $r$, the cost of any treasure hunt algorithm using advice of size $z$ must be at least $\frac{1}{16}\cdot \frac{x^2}{2^z r}$. This proves the claim.

For any positive integer $i$, consider the squares $Q_{j}$ with side $2^{j+1}$, for $j=0,1,2,\dots,i$, centered at $P$, with sides vertical or horizontal. Denote $R_j = Q_j \setminus Q_{j-1}$, for $j=1,2,\dots, i$. Consider the sub-square $K_j$ of $R_j$ such that the North-West corner of $K_j$ is the same as the North-West corner of $Q_j$ and the South-East corner of $K_j$ is the same as the North-West corner of $Q_{j-1}$. For each square $K_j$,  consider couples of integers $(D_j,r_j)$, where $D_j=\sqrt{2}\cdot 2^j$ and $r_j=\frac{1}{2^{2(i-j)}}$, for $j=1,2,\dots, i$.

Notice that all the points of the square $K_j$ are at distance at most $D_j$ from the point $P$,  and the side of the square $K_j$ is $2^j/2$. Suppose that the target is hidden in a square $K_j$, with vision radius $r_j$, for $j=1,2, \dots, i$. By the claim, for each square $K_j$, the cost of any treasure hunt algorithm using advice of size $z$ is at least $\frac{1}{16}\cdot \frac{2^{2j}}{4\cdot 2^z r_j} = \frac{1}{64} \cdot \frac{2^{2j}}{2^z r_j}$. Since $r_j=\frac{1}{2^{2(i-j)}}$,  the cost is at least $\frac{1}{64}\cdot \frac{2^{2i}}{2^z}$. Since squares $K_j$ are pairwise disjoint, the cost of any algorithm that accomplishes treasure hunt in every square $K_j$, using advice of size $z$, is at least $\frac{1}{64}\cdot i \cdot \frac{2^{2i}}{2^z}$. Since $D_j=\sqrt{2}\cdot 2^j$ and $r_j=\frac{1}{2^{2(i-j)}}$, we have $\frac{D_j^2}{r_j} = 2\cdot 2^{2j}\cdot 2^{2(i-j)} = 2\cdot 2^{2i}$. Since $D_j=\sqrt{2}\cdot 2^j$, we have $\log D_j = \frac{1}{2}+j$. Since $r_j=\frac{1}{2^{2(i-j)}}$, we have $\log \frac{1}{r_j} = 2(i-j) $.

Hence, we have $\frac{D_j^2}{r_j} (\log D_j + \log \frac{1}{r_j}) = 2\cdot 2^{2i}\cdot (\frac{1}{2}+j + 2(i-j))$. Since $j > 1/2$, we have $ 2\cdot 2^{2i}\cdot (\frac{1}{2}+j + 2(i-j)) \leq 2\cdot 2^{2i}\cdot 2i$. Hence, we have $2^{2i}\cdot  i \geq \frac{1}{4}\cdot \frac{D_j^2}{r_j} (\log D_j + \log \frac{1}{r_j})$. Hence, the cost of the algorithm is at least $\frac{1}{64}\cdot i \cdot \frac{2^{2i}}{2^z} \geq \frac{1}{64}\cdot \frac{1}{2^z} \cdot \frac{1}{4}\cdot \frac{D_j^2}{r_j} (\log D_j + \log \frac{1}{r_j})= \frac{1}{256} \cdot \frac{D_j^2}{2^z r_j}(\log D_j + \log \frac{1}{r_j})$. Together with the previously observed lower bound $\Omega(D)$, this implies that the cost of any treasure hunt algorithm using advice of size $z$ is $\Omega(D+\frac{D^2}{2^zr}(\log D + \log 1/r))$. 
\end{proof}

Theorems \ref{small-vision-cost} and \ref{small-lb} imply the following corollary.

\begin{corollary}
Suppose that the treasure is at distance at most $D$ from the initial position of the agent and the vision radius $r$  is at most 1.
The cost of Algorithm {\tt Small vision}, using advice of size $z$ is $O(OPT(z,D,r))$, where $OPT(z,D,r)$ is the cost of the optimal algorithm using advice of size $z$.
\end{corollary}

\section{Medium vision radius}
In this section, we consider the case of medium vision radius, i.e. $1 < r < 0.9D$.  Our aim is to design an almost optimal scheme of treasure hunt algorithms using canonical
advice of size $z\geq 0$: for any fixed real $\alpha >0$, which is given as input, the cost of the algorithm will be  $O(OPT(z,D,r)^{1+\alpha})$, where $OPT(z,D,r)$ is the cost of the optimal algorithm using advice of size $z$ (for unknown upper bound $D$ on the distance from the initial position of the agent to the treasure, and unknown vision radius $r$).
Let $C(z,D,r)$ be the cost of Algorithm {\tt Basic Traversal} with input $D$ and $r$, and using advice of size $z$. 
 For a given constant $\alpha>0$, our algorithm (not knowing $D$ or $r$, and only using advice of size $z$) will work at cost $O(C(z,D,r)D^{\alpha})$.
 We will show that this cost is in
 $O(OPT(z,,D,r)^{1+\alpha})$, as desired.

Let $z$ be a non-negative integer.
Fix a constant $\alpha > 0$, Let  $c=\lceil 1/\alpha \rceil$. 
Consider the infinite matrix $A$ whose rows and columns are numbered by consecutive positive integers. Let $A[i,j]$ be the entry in the $i$th row and $j$th column of $A$.  
Let $s$ be the smallest integer such that $2^s \geq 2 \cdot 2\cdot 138\cdot 800$. Hence $s=20$.  The entry $A[i,j]$ of the matrix $A$ is defined as the execution of 
Algorithm {\tt Basic Traversal} with parameters $D'=2^{js}$ and $r'=2^i$, and using advice of size $z$.
We only consider entries $A[i,j]$ such that $i\leq js$
because we assume that the radius vision does not exceed the distance between the treasure and the initial position of the agent.
Hence, for any $j$, we define the $j$th column of $A$ as the sequence of entries $A[1,j],A[2,j],\dots,A[js,j]$.
   
 Consider the $j$th column of the matrix $A$ such that $j \geq c/s$. This column has $js\geq c$ entries. Let $p$ and $q$ be the two positive integers such that $js = p\lfloor js/c \rfloor +q \lceil j/sc \rceil $ and $c=p+q$. 
Partition the $j$th column into $c$ segments, $p$ of them of length $\lfloor js/c \rfloor$ and $q$ of them of length $ \lceil js/c \rceil$. The $p$ segments of 
length $\lfloor js/c \rfloor$ cover entries $A[1,j]$ to $A[p\lfloor js/c \rfloor,j]$ and the $q$ segments of length $ \lceil js/c \rceil$ cover the rest of the $j$th column. 
Next, we define the notion of a {\em  dot} in a given column. A dot is the lowest-indexed entry of each segment. Hence, any column $j \geq c/s$ contains $c$ dots.

Let $x=\lfloor \frac{js}{c} \rfloor$. The first $p$ dots in the $j$th column are entries $A[(l x +1),j]$, for $l=0,1,\dots (p-1)$, and the next $q$ dots in the $j$th column are the entries $A[(px+m(x+1)+1),j]$, for $m=0,1, \dots (q-1)$.


For any $k\leq  c$, define $Thread_k$ to be the sequence consisting of the $k$th dot in every column $j\geq c/s$ of the matrix $A$ (counting from 1 in the order of increasing row indices). For example, $Thread_1$ is the sequence of dots $A[1,j]$, for $j\geq c/s$, cf. Fig. \ref{matrix}). 
We will consider the entries of each thread in the order of increasing column indices $j\geq c/s$. 

Our scheme of treasure hunt algorithms is formalized as a single Algorithm {\tt Medium vision} which, apart from the canonical advice, gets as input a real parameter $\alpha>0$ that  will determine how close to optimal is the algorithm cost.
The high-level idea of Algorithm {\tt Medium vision} is to fill the dots in a carefully chosen order, where filling a dot  $A[i,j]$ means executing Algorithm {\tt Basic Traversal} 
with advice $z$ and inputs $D'=2^{js}$ and $r'=2^i$, and backtracking to the initial position $P$ using the reverse trajectory. Dots are filled in this order until the treasure is found. We will show that this happens at the latest, at the time when a particular dot, depending on the unknown parameters $D$ and $r$, is filled. This dot will be called {\em special}. The order of filling the dots is chosen in such a way that  the total cost incurred until the special dot is filled approximates well the optimal cost $OPT(z,D,r)$. The factor $\Theta(D^\alpha)$ that separates our solution from the optimal cost is due to the fact that instead of executing Algorithm {\tt Basic Traversal} for parameters $D$ and $r$, we execute it for some parameters $D^*$ and $r^*$ corresponding to the special dot, where $D^*$ is approximately $D$ and $0 < r- r^*$ is approximately $\alpha \log D$.

We now give a detailed description of the algorithm. The algorithm works in phases. Each phase is started by filling the first non-filled dot in $Thread_c$. Suppose that this dot is $A[i,j]$. Let $D_j=2^{js}$ and $r_i=2^i$. The agent gets the {\em budget} $B(z,D_j,r_i)=2 \cdot 138\cdot(\frac{D_j^2}{2^zr_i}+D_j)$ for each dot in this phase.
The agent fills the dot $A[i,j]$ (which is within the budget) and tries to fill the yet unfilled dots in each $Thread_k$, for $k<c$, as follows.
Let $j_k$ be the column such that the first yet unfilled dot in $Thread_k$ is in column $j_k$. The agent tries to fill consecutive dots in $Thread_k$, starting from column $j_k$ in order of increasing columns. Whenever filling a given dot $A[i',j']$ is within the budget, i.e., $2 
\cdot C(z,D_{j'},r_{i'})\leq B(z,D_j,r_i)$,  the agent fills  the given dot. Otherwise, the dot remains unfilled in this phase.  (Note that the agent can compute $C(z,D_{j'},r_{i'})$ for any parameters $D_{j'},r_{i'}$ by simulating the execution of Algorithm {\tt Basic Traversal}, for these inputs and for advice of size $z$).  This ends the phase. The first phase starts by filling the first dot in $Thread_c$.

\begin{figure}[tp]
\centering
\includegraphics[scale=0.6]{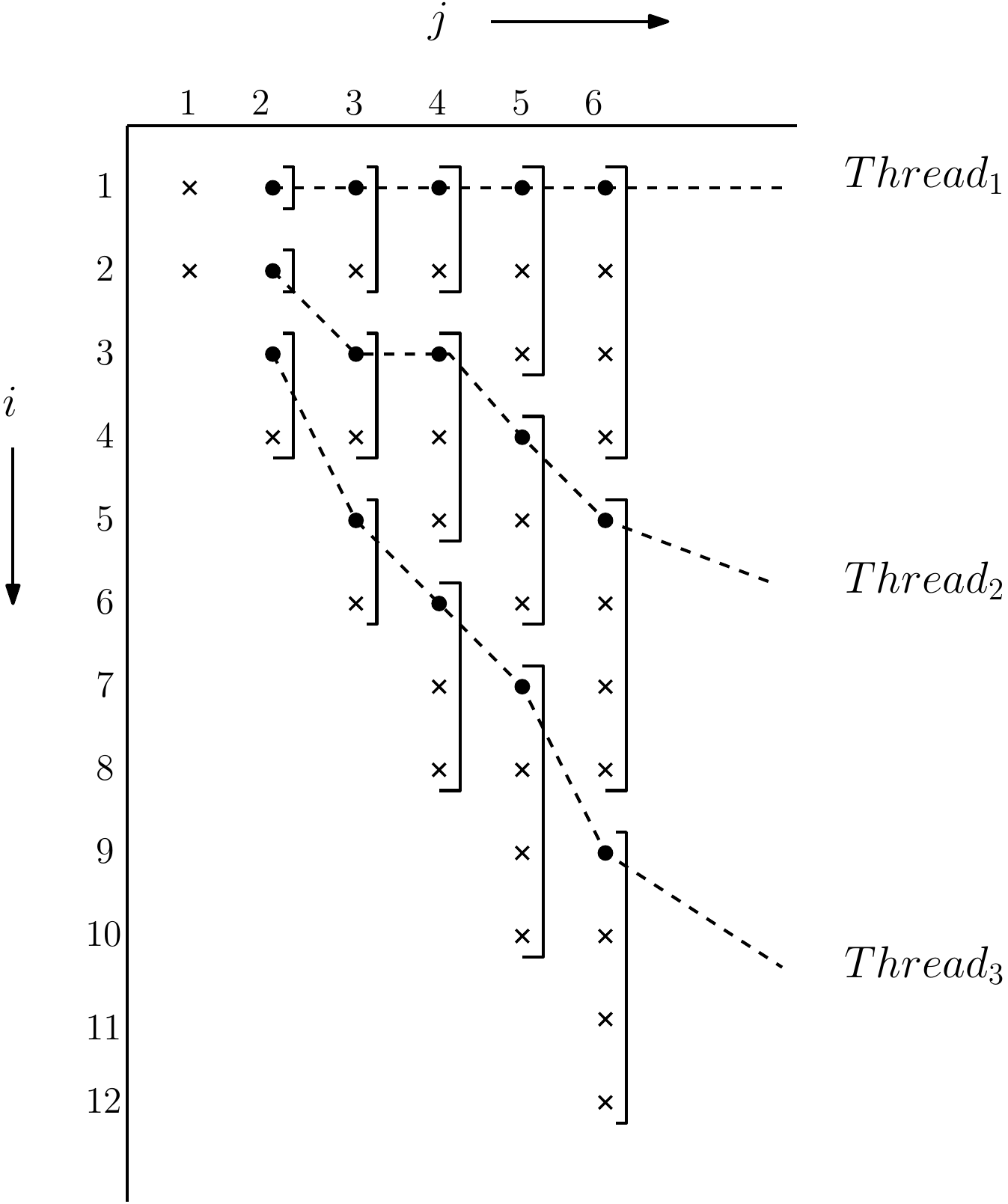}
\caption{The matrix $A$, dots and threads, for $c=3$ and $s=2$.}
\label{matrix}
\end{figure}

We formulate the algorithm {\tt Medium vision} as follows, using the notion of filling dots in the matrix $A$, described above. 
The input of the algorithm is a positive real $\alpha$, and the algorithm works with advice of size $z$.
The algorithm is interrupted when the agent sees the treasure.

\begin{center}
\fbox{
\begin{minipage}{15cm}

\noindent
{\bf Algorithm} {\tt Medium vision}

\noindent
$p:=1$\\
$c:=\lceil 1/\alpha \rceil$\\
 {\bf repeat}\\
\hspace*{1cm}Let $A[i,j]$ be the first unfilled dot in $Thread_c$. Let  $B(z,D_j,r_i)=2 \cdot138\cdot(\frac{D_j^2}{2^zr_i}+D_j)$\\
\hspace*{1cm}Fill the dot $A[i,j]$\\
\hspace*{1cm}{\bf for} $k:=(c-1)$ {\bf down to} 1 {\bf do}\\
\hspace*{2cm}$t:=0$\\
\hspace*{2cm}Let $j_k$ be the column number of the first yet unfilled dot in $Thread_k$.\\
\hspace*{2cm}Let $i_t$ be the row number of the dot in $Thread_k$ and column $j_k+t$.\\
\hspace*{2cm}{\bf while} $2 \cdot C(z,2^{(j_k+t)s},2^{i_t}) \leq B(z,D_j,r_i)$ {\bf do}\\
\hspace*{3cm}fill dot $A[i_t,j_k+t]$\\
\hspace*{3cm}$t:=t+1$\\ 
\hspace*{1cm}$p:= p+1$

\end{minipage}
}
\end{center}

We first prove the correctness of Algorithm {\tt Medium vision}.

\begin{lemma}
Let $1 < r < 0.9D$, where the treasure is at distance at most $D$ from the initial position of the agent and $r$ is the vision radius.
 Algorithm {\tt Medium vision} correctly finds the treasure.
\end{lemma}

\begin{proof}
Let $j$ be the smallest integer such that $2^{js} \geq D$ and let $i$ be the largest integer such that $2^i \leq r$. Filling the dot $A[i,j]$ corresponds to executing
Algorithm {\tt Basic Traversal} with parameters $D'=2^{js}$, and $r' = 2^i$, using canonical advice of size $z$, and backtracking. 
Since in consecutive phases $p$, consecutive dots in $Thread_c$ are filled, the budget available in phase $p$ grows to infinity with $p$. Hence, for some phase $p$, the budget is sufficient to fill dot $A[i,j]$. 
Since $D'\geq D$ and   
$r'\leq r$, the agent must see the treasure by the end of the execution of the Algorithm {\tt Basic Traversal} corresponding to filling this dot.  
\end{proof}

The following lemma shows that the cost of filling consecutive dots of a given thread grows exponentially.

\begin{lemma}\label{telescope}
For a given positive integer $k \leq c$, and a given integer $d>1$, the cost of filling the $d$th dot of $Thread_k$ is at least 2 times larger than the cost of filling the $(d-1)$th dot of $Thread_k$. 
\end{lemma}
\begin{proof}
Let $A[i,j]$ be the entry corresponding to the $(d-1)$th dot of $Thread_k$. The corresponding execution of Algorithm {\tt Basic Traversal} uses parameters
$D_{j}=2^{js}$, $r_i= 2^i$ and canonical advice of size $z$.

Let $A[i',j']$ be the entry corresponding to the $d$th dot of $Thread_k$. Since $j'=j+1$ and $i'\leq i+s$, the Algorithm {\tt Basic Traversal} is executed using parameters $D_{j'}=2^{(j+1)s} = 2^s\cdot 2^{js}$, and $r_{i'} \leq 2^{(i+s)}=2^s\cdot 2^i$ with advice of size $z$, during the filling of the $d$th dot of $Thread_k$. 
Since $\frac{D_j}{2^z r_i} = \frac{2^{js}}{2^z 2^i}$ and  $\frac{D_{j'}}{2^z r_{i'}} \geq \frac{2^s \cdot 2^{js}}{2^z \cdot 2^i \cdot 2^s}$, we have $\frac{D_{j'}}{2^z r_{i'}} \geq \frac{D_j}{2^z r_i}$. Hence, we have $\frac{D_{j'}}{2^z r_{i'}}+1 \geq \frac{D_j}{2^z r_i} +1$, and thus $(\frac{D_{j'}^2}{2^zr_{i'}}+D_{j'})/ (\frac{D_j^2}{2^zr_i}+D_j) \geq D_{j'}/ D_j = 2^s$. Hence, we have $(\frac{D_{j'}^2}{2^zr_{i'}}+D_{j'}) \geq 2^s(\frac{D_j^2}{2^zr_i}+D_j)$. 
By Lemma \ref{cost-basic}, we have  $C(z,D_j,r_i)\leq 138\cdot (\frac{D_j^2}{2^zr_i}+D_j)$, where $D_j= 2^{js}$ and $r_i=2^i$.
By Lemma \ref{lb-medium}, we have $C(z,D_{j'},r_{i'})\geq \frac{1}{800}(\frac{D_{j'}^2}{2^zr_{i'}}+D_{j'})$. Since  $(\frac{D_{j'}^2}{2^zr_{i'}}+D_{j'}) \geq 2^s(\frac{D_j^2}{2^zr_i}+D_j)$, we have $C(z,D_{j'},r_{i'}) \geq \frac{1}{800}(\frac{D_{j'}^2}{2^zr_{i'}}+D_{j'}) \geq \frac{2^s}{800}(\frac{D_j^2}{2^zr_i}+D_j)$. Since $C(z,D_j,r_i) \leq 138\cdot (\frac{D_j^2}{2^zr_i}+D_j)$, we have $C(z,D_{j'},r_{i'}) \geq \frac{2^s}{800} \cdot \frac{1}{138} \cdot C(z,D_j,r_i)$. Since $2^s \geq 2\cdot 138 \cdot 800$, we have $C(z,D_{j'},r_{i'}) \geq 2 \cdot C(z,D_j,r_i)$. This proves the lemma.
\end{proof}


For any dot $\Delta$, let $G(\Delta)$ denote the cost of filling this dot by Algorithm {\tt Medium vision}.
The next lemma shows that Algorithm {\tt Medium vision} fills dots in a cost-efficient order with respect to the cost of filling any given dot (up to multiplicative constants).

\begin{lemma}\label{phase-cost}
Let $\Delta$ be any dot that is filled by Algorithm {\tt Medium vision} in phase $p$. Then the cost of Algorithm {\tt Medium vision} until the end of phase $p$ is 
at most $2c \cdot 2^{2s}G(\Delta)$. 
\end{lemma}
\begin{proof}
Let $\Delta_p$ be the dot in $Thread_c$ filled in phase $p$ of Algorithm {\tt Medium vision}. Suppose that the dot $\Delta_p$ is the entry $A[i,j]$ of the matrix.
Hence the budget $B$ available in phase $p$ is $2 \cdot 138\cdot(\frac{D_j^2}{2^zr_i}+D_j)$.
Let $\Delta_{p-1}$ be the dot in $Thread_c$ filled in phase $p-1$ of Algorithm {\tt Medium vision}. Suppose that the dot $\Delta_{p-1}$ is the entry $A[i',j']$ of the matrix. Hence the budget $B'$ available in phase $p-1$ is $2 \cdot 138\cdot(\frac{D_{j'}^2}{2^zr_{i'}}+D_{j'})$. 
Since the dot $\Delta$ is filled only in phase $p$, the cost $G(\Delta)$ of filling it must be larger than the budget $B'$, i.e. $B' < G(\Delta)$. 
By definition of $B$ and $B'$, we have 
$B \leq 2^{2s} \cdot B'$. Hence $B \leq 2^{2s} B' < 2^{2s} G(\Delta)$. 

Let $\delta_k$, for $k\leq c$, be the last dot in $Thread_k$ filled in phase $p$. The cost of filling this dot must be within the budget of phase $p$. Hence 
$G(\delta_k)\leq  B$. By Lemma \ref{telescope}, the cost of filling all dots in $Thread_k$ by the end of phase $p$ is at most $2G(\delta_k)$. Hence the total cost 
of Algorithm {\tt Medium vision} until the end of phase $p$ is at most $2c \cdot 2^{2s}G(\Delta)$.
\end{proof}

The following lemma estimates the cost of filling the special dot.

\begin{lemma}\label{special-dot}
Let $1 < r < 0.9D$, where the treasure is at distance at most $D$ from the initial position of the agent and $r$ is the vision radius.
Fix a size $z\geq 0$ of advice and fix an input $\alpha>0$ of Algorithm {\tt Medium vision}. 
Let $C(z,D,r)$ be the cost of Algorithm {\tt Basic Traversal} with input $D$ and $r$, and using advice of size $z$. 
Then the cost of filling the special dot by Algorithm {\tt Medium vision} is at most $2^{5s} \cdot C(z,D,r)D^\alpha$, where $s=20$.
\end{lemma}

\begin{proof}
Let $j$ be the smallest integer such that $2^{js} \geq D$ and let $i$ be the largest integer such that $2^i \leq r$. Let $A[i,j]$ be the entry of the matrix $A$ corresponding to integers $i$ and $j$. Since $ 2^{(j-1)s} <D  \leq 2^{js}$, we have $2^{js} \leq 2^s D$. Since $ 2^i \leq r < 2^{i+1}$, we have $2^i \geq r/2$. The cost $C(z,2^{js},2^i)$ of Algorithm {\tt Basic Traversal} with inputs $2^{js}$ and $2^i$, and using advice of size $z$, is at most $138\cdot(\frac{(2^{js})^2}{2^z 2^i}+2^{js})$, in view of Lemma \ref{cost-basic}. Since $2^{js} \leq 2^s D$ and $2^i \geq r/2$, we have $C(z,2^{js},2^i) \leq 138 \cdot (\frac{2 \cdot 2^{2s} D^2}{2^z r}+2^s D) \leq 2 \cdot 2^{2s} \cdot 138 (\frac{D^2}{2^z r} +D)$. By Lemma \ref{lb-medium}, we have $C(z,D,r) \geq \frac{1}{800}(\frac{D^2}{2^z r}+D)$. Hence $C(z,2^{js},2^i) \leq 2 \cdot 2^{2s} \cdot 138 \cdot 800 \cdot C(z,D,r)$. Since $2^s \geq 2\cdot 138 \cdot 800$, we have $C(z,2^{js},2^i) \leq 2^{3s} C(z,D,r)$.

Let $t$ be the largest integer such that $t \leq i$ and  $A[t,j]$ is a dot in column $j$ of the matrix $A$. This is the special dot with respect to parameters $D$ and $r$.  Call this dot $S$. Since any segment in the column $j$ can have at most $\lceil \frac{js}{c} \rceil$ entries, we have $i \leq t+ \lceil\frac{js}{c}\rceil$. The cost $G(S)$ of filling the dot $S$ is at most $2\cdot 138\cdot (\frac{(2^{js})^2}{2^z 2^t}+2^{js})$. Since $t \geq i-\lceil\frac{js}{c}\rceil \geq i-\frac{js}{c}-1$, we have $G(S) \leq 2 \cdot 138\cdot (\frac{(2^{js})^2}{2^z 2^{(i-\frac{js}{c}-1)}}+2^{js}) \leq 2 \cdot 138\cdot 2 \cdot 2 ^{\frac{js}{c}}(\frac{(2^{js})^2}{2^z 2^i}+2^{js})$. Since $C(z,2^{js},2^i) \geq \frac{1}{800}(\frac{(2^{js})^2}{2^z  2^i} + 2^{js})$, we have $G(S) \leq 2 \cdot 2 \cdot 138\cdot 800 \cdot 2 ^{\frac{js}{c}}\cdot C(z,2^{js},2^i)$. Since $2^s \geq 2\cdot 2\cdot 138 \cdot 800$ and $C(z,2^{js},2^i) \leq 2^{3s} \cdot C(z,D,r)$, we have $G(S) \leq 2^{4s} \cdot C(z,D,r)\cdot 2 ^{\frac{js}{c}}$. Since $2^{js} \leq 2^s D$, we have $G(S) \leq 2^{4s} \cdot 2^{s/c}\cdot C(z,D,r)\cdot D^{1/c}$. Since $c = \lceil\frac{1}{\alpha}\rceil$, we have $1/c\leq{\alpha}$, and thus  $G(S) \leq 2^{4s} \cdot 2^{s/c}\cdot C(z,D,r)\cdot D^\alpha$. Hence $G(S) \leq 2^{5s}\cdot C(z,D,r)\cdot D^\alpha$. This proves the lemma. 
\end{proof}

We are now able to estimate the cost of Algorithm {\tt Medium vision}. Recall that this algorithm works with canonical advice of size $z\geq 0$
and uses as input an arbitrary positive real constant $\alpha$. This is the only knowledge available to the agent. We show that the cost of this algorithm approximates the cost of the optimal algorithm using advice of size $z$.

\begin{theorem}\label{medium-vision-cost}
Let $1 < r < 0.9D$, where the treasure is at distance at most $D$ from the initial position of the agent and $r$ is the vision radius.
 Algorithm {\tt Medium vision}, using advice of size $z$ and any positive real input $\alpha$, works  at cost ${\cal C}(\alpha ,z,D,r)\in O(C(z,D,r) D^{\alpha})$
 which is
 $O(OPT(z,D,r)^{1+\alpha})$, where 
  $OPT(z,D,r)$ is the cost of the optimal algorithm using advice of size~$z$.
\end{theorem}

\begin{proof}
In view of Lemmas \ref{telescope}, \ref{phase-cost} and \ref{special-dot}, we have
${\cal C}(\alpha ,z,D,r)\leq \lceil 1/{\alpha}\rceil \cdot 2^{141} \cdot C(z,D,r) D^{\alpha}$.
In view of Lemma \ref{cost-basic}, we have $C(z,D,r) \leq 138\cdot(\frac{D^2}{2^zr}+D)$.
By Lemma \ref{lb-medium}, we have $OPT(z,D,r) \geq \frac{1}{800}(\frac{D^2}{2^zr}+D)$.
Hence $C(z,D,r)D^{\alpha}\leq 800\cdot 138 \cdot OPT(z,D,r) \cdot D^{\alpha}$. 
Since $OPT(z,D,r) \geq \frac{1}{800}D$, we have $C(z,D,r)D^{\alpha}\leq 800\cdot 138 \cdot OPT(z,D,r) \cdot (800\cdot OPT(z,D,r))^{\alpha}$.
Hence ${\cal C}(\alpha ,z,D,r)\in O(OPT(z,D,r)^{1+\alpha})$.
\end{proof}

\section{Large vision radius}

The obvious lower bound on the cost of  treasure hunt is $D-r$, as the treasure can be at distance exactly $D$ from the initial position of the agent.
Notice that $D-r$ can be much smaller than $D$, e.g., in the case when $r=D-\log D$.
In this section we consider the case of large $r$, more precisely when $r \geq 0.9D$. In this case we will design a treasure hunt algorithm working, without any advice, at cost $O(D-r)$, and hence optimal
(up to multiplicative constants).

Let $L_i$, for $i=0,\dots, 12$, be the half-lines starting at the initial position $P$ of the agent and forming angle $\pi i/12$ with direction North, counterclockwise from this direction. Thus $L_0=L_{12}$ is in direction North.  Let $S_i$, for $i=0,\dots, 11$,  be the sector between lines $L_i$ and $L_{i+1}$. Our algorithm can be formulated as follows. It is interrupted when the agent gets at distance $r$ from the treasure.

\begin{center}
\fbox{
\begin{minipage}{10cm}

\noindent
{\bf Algorithm} {\tt Large vision}

\noindent
$j:=1$\\
{\bf repeat}\\
\hspace*{1cm}{\bf for} $i:=0$ {\bf  to} 11 {\bf do}\\
\hspace*{2cm}go along line $L_i$ at distance $2^j$ and go back to $P$\\ 
\hspace*{1cm}$j:= j+1$
\end{minipage}
}
\end{center}

In the analysis of the algorithm, we will use the following technical lemma.

\begin{lemma}\label{bound}
Consider reals $D,r$ such that $0.9D \leq r < D$. Let $PR$ be a line segment of length $D$ and let $S$ be a point in this segment such that $|PS|=D-r$ and $|SR|=r$. Let $L$ be a half-line starting from point $P$, forming an angle $\pi/6$ with the segment $PR$. Let $Q$ be the point in $L$ closest to $P$
such that $|QR|=r$. Then $|PQ| \leq 1.2 \cdot |PS|$.
\end{lemma}

\begin{figure}[tp]
\centering
\includegraphics[scale=1.0]{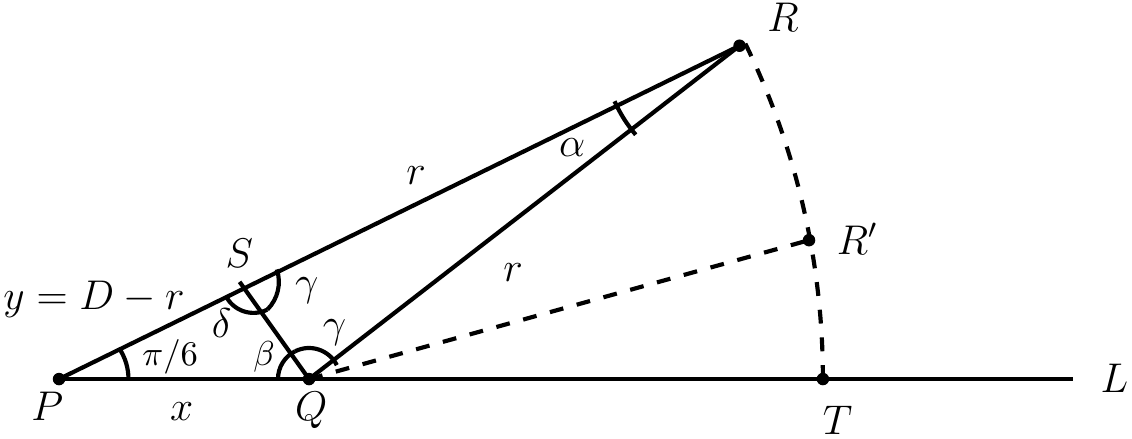}
\caption{Illustration of the proof of Lemma \ref{bound} }
\label{long-range}
\end{figure}

\begin{proof}
Let $|PQ|=x$ and $|PS|=y$. Denote the angle $\angle PSQ$ by  $\delta$, the angle $\angle SQR$ by $\gamma$, the angle $\angle QRS$ by $\alpha$ and the angle $\angle PQS$ by $\beta$, see Fig. \ref{long-range}. 
 
By the sine rule applied to the triangle $PQR$, we have $\frac{\sin (\pi/6)}{\sin (\beta+\gamma)}=\frac{r}{D}$. Since $r/D \geq 0.9$, we have $\sin(\beta+\gamma) \leq \sin (\pi/6)/0.9 \leq 0.56$. We have arcsin 0.56 = 0.594.
Let $A=(\pi-0.594) \geq 2.54$. Since $Q$ is the point in $L$ closest to $P$ such that $|QR|=r$, we have $(\beta+\gamma) > \pi/2$. Thus, by definition of $A$ we have $(\beta+\gamma) \geq A$. Using the triangle $PQR$, we have $\alpha = \pi-\pi/6-(\beta+\gamma) \leq \frac{5\pi}{6}-A$. Since the triangle $SQR$ is isoceles, we have $\gamma=(\pi-\alpha)/2$. Since $\alpha \leq \frac{5\pi}{6}-A $, we have $\gamma \geq \frac{\pi}{12}+\frac{A}{2}$. By definition, we have $\delta = \pi-\gamma$. Since $\gamma \geq \frac{\pi}{12}+\frac{A}{2}$, we have $\delta \leq \frac{11 \pi}{12}-\frac{A}{2}$. Using the triangle $PQS$, we have $\beta = (\pi-\frac{\pi}{6}-\delta)=\frac{5\pi}{6}-\delta$. Since $\delta \leq \frac{11 \pi}{12}-\frac{A}{2}$, we have $\beta \geq \frac{A}{2}-\frac{\pi}{12} \geq \frac{2.54}{2}-\frac{\pi}{12}\geq 1$.

By the sine rule applied to the triangle $PQS$, we $\frac{x}{y}=\frac{\sin \delta}{\sin \beta} \leq \frac{1}{\sin\beta}$. Since $\pi/2\geq \beta \geq 1$, we have 
$\sin \beta \geq \sin 1$ and thus
$\frac{x}{y}\leq \frac{1}{\sin 1} \leq 1.2$. This proves the lemma.
 \end{proof}
 
 We are now able to prove the correctness and estimate the cost of Algorithm {\tt Large vision}.
 
 \begin{theorem}\label{large vision}
 Suppose that the treasure is at distance at most $D$ from the initial position of the agent and that the vision radius $r$ satisfies $D>r \geq 0.9D$, where parameters $D$ and $r$ are unknown to the agent. Then an agent executing 
 Algorithm {\tt Large vision} finds the treasure at cost $O(D-r)$.
 \end{theorem}
 
 \begin{proof}
 Suppose that the  treasure is at distance  at most $D$ from the initial position $P$ of the agent. Let $j_0=\lceil \log (1.2 \cdot (D-r)) \rceil$. We will show that the agent sees the treasure by the end of the execution of the {\bf repeat} loop for $j=j_0$. 
 
 Suppose that the treasure is in sector $S_i$, where $0\leq  i \leq 11$. Let ${S_i}^*$ be the set of points which is the intersection of the sector $S_i$ with the disc of radius $D$ centered at $P$.
 Let $R$ be the point on line $L_{i+1}$ at distance $D$ from $P$, and let  $Q$ be the point on line $L_i$, closest to $P$ such that $|RQ|=r$. By Lemma \ref{bound}, we have $|PQ| \leq 1.2 \cdot (D-r)$. 
 
We will show that the agent sees all the points in the set ${S_i}^*$, when it reaches the point $Q$. Let ${S_i}'$ be the set of points in the triangle $PQR$.  We first show that $|PQ| \leq r$. Since $r \geq 0.9D$, we have $D-r \leq 0.1D$, and hence $|PQ| \leq 1.2(D-r) \leq 1.2 \cdot 0.1 D \leq 0.12D < r$. Since the points $R$ and $P$ are at distance at most $r$ from the point $Q$, the triangle $PQR$ must be inside the circle of radius $r$ centered at point $Q$. Hence, the agent sees all the points in the set ${S_i}'$ from point $Q$. Next, we show that the agent can also see all the points in the set ${S_i}^*-{S_i}'$ from the point $Q$. 

Let $R'$ be any point on the arc $RT$, see Fig. \ref{long-range}. Let $|QR'|=x'$. Consider the triangle $PR'R$. Since $|PR|=|PR'|=D$, the triangle $PR'R$ is an isosceles triangle. In view of Lemma \ref{geom},
since $|QR|=r$, we have $r > x'$. Hence, from the point $Q$, the agent can see all the points in the arc $RT$, and hence also all the points in the set ${S_i}^*-{S_i}'$. It follows that when the agent reaches point $Q$, it sees all the points in  ${S_i}^*$, and thus it must see the treasure by this time. This proves
that the agent sees the treasure by the time it gets to point $Q$.

We finally show that the agent reaches the point $Q$ by the end of the execution of the {\bf repeat} loop for $j=j_0$ of Algorithm {\tt Large vision}.  
During this execution of the loop, the agent goes at distance $2^{j_0}= 2^{\lceil \log 1.2 (D-r)\rceil} \geq 1.2 (D-r)$ from point $P$ along line $L_i$, and hence, it reaches point $Q$ by the end of this execution.

We compute the cost of Algorithm {\tt Large vision} as follows.  The cost of the $j$th execution of the {\bf repeat} loop is $12\cdot 2 \cdot 2^j$.
Hence the cost of the first $j_0$ executions of the loop is at most $2\cdot 12\cdot 2 \cdot 2^j= 48\cdot 2^{j_0}$.
Since $2^{j_0} \leq 2.4 (D-r)$, the cost of our algorithm is at most $48\cdot 2^{j_0} \leq 48\cdot 2.4(D-r)< 116(D-r)$. This concludes the proof.  
  \end{proof}
  
  Theorem \ref{large vision} implies the following corollary.
  
  \begin{corollary}
  Suppose that the treasure is at distance at most $D$ from the initial position of the agent and the vision radius $r$ is  at least $0.9D$.
The cost of Algorithm {\tt Large vision} (working without any knowledge or advice) is $O(OPT(z,D,r))$, where $OPT(z,D,r)$ is the cost of the optimal algorithm using advice of any size $z$.
   \end{corollary}

 \section{The universal algorithm}
 
 In the previous sections we presented three treasure hunt algorithms working in three different ranges of the value of vision radius $r$. However, using any of them requires knowing that the vision radius is in a given range. We need a universal treasure hunt algorithm that would work efficiently for any (unknown) values of $D$ and $r$. One way to  design such an algorithm from our previously constructed building blocks would be to reserve 2 bits of the advice to indicate in which of the three cases the agent is situated, and then use $z-2$ bits for the actual advice, as described before. This would work for any size of advice $z\geq 2$, as $z$ and $z-2$ are of the same order of magnitude, hence using $z-2$ instead of $z$ bits for the ``real'' advice would not change the complexity of the solution. However, in the case $z\leq 1$, which includes the important case of no advice whatsoever, this solution does not work. Hence we present an alternative way of ``merging'' our three algorithms into one universal algorithm working for any size $z \geq 0$ of advice. 
 
 Fix any non-negative integer $z$, and let 
  $\Sigma_1$, $\Sigma_2$, $\Sigma_3$ be the trajectories of the agent resulting from executing, respectively, Algorithm {\tt Small vision},
  Algorithm {\tt Medium vision} and Algorithm {\tt Large vision}, using the canonical advice of size $z$. In the case of   Algorithm {\tt Medium vision},
  we use an arbitrary fixed constant $\alpha >0$ as input of the algorithm, and the trajectory $\Sigma_2$ is obtained by executing this algorithm with this input.
  In the case of   Algorithm {\tt Large vision}, any advice should be ignored, as the algorithm works without it. Each of the trajectories $\Sigma_i$ is an infinite polygonal line in the plane, starting at the initial position $P$ of the agent.
  Using the obtained advice, the agent can compute each of the above trajectories to an arbitrary finite length.
  In each of the above three algorithms, the agent follows the respective trajectory until it finds the treasure.
  
  The idea of merging the three algorithms (without knowing which of the ranges of vision radius $r$ is the actual one) is to follow each of the three trajectories in a round-robin fashion, at  distances
  increasing exponentially, each time backtracking to point $P$. Below is the pseudocode of the algorithm using canonical advice of some non-negative size $z$. The algorithm has a positive real input $\alpha$, used in the construction of trajectory $\Sigma_2$.  As usual, the algorithm is interrupted when the
  agent sees the treasure.
  
  \begin{center}
\fbox{
\begin{minipage}{9cm}

{\bf Algorithm} {\tt Universal}\\

\noindent
$p:=1$\\
{\bf repeat}\\
 \hspace*{1cm}Go at distance $2^p$ along the trajectory $\Sigma_1$\\
 \hspace*{1cm}Backtrack to the initial position $P$\\
 \hspace*{1cm}Go at distance $2^p$ along the trajectory $\Sigma_2$\\
 \hspace*{1cm}Backtrack to initial position $P$\\
  \hspace*{1cm}Go at distance $2^p$ along the trajectory $\Sigma_3$\\
 \hspace*{1cm}Backtrack to initial position $P$\\
 \hspace*{1cm}$p:=p+1$\\
\end{minipage}
}
\end{center}

The following theorem establishes the correctness and estimates the cost of Algorithm {\tt Universal}.
For any size $z$ of advice and any $D$ and $r$, let $OPT(z,D,r)$ be the optimal cost of finding the treasure for parameters $z$, $D$ and $r$, if the agent has only an advice string of length $z$ as input.

\begin{theorem}
Fix any constant $\alpha >0$. 
For any size $z$ of advice and any $D$ and $r$,
Algorithm {\tt Universal} works with advice of size $z$ at cost $O(OPT(z,D,r))$ whenever $r\leq 1$ or $r\geq 0.9D$. For intermediate values of $r$, i.e., for $1<r<0.9D$, it works at cost $O(OPT(z,D,r)^{1+\alpha})$, where $\alpha$ is used as input 
to determine the trajectory $\Sigma_2$.
\begin{itemize}
\item
If $r\leq 1$, the cost of Algorithm {\tt Universal} is $O(D+\frac{D^2}{2^zr}(\log D + \log 1/r))$. 
\item
If $1<r<0.9D$, the cost of Algorithm {\tt Universal} is $O((D+\frac{D^2}{2^zr})D^{\alpha})$.
\item
If $r\geq 0.9D$, the cost of Algorithm {\tt Universal} is $O(D-r)$.
\end{itemize}
\end{theorem}

\begin{proof}
Let $x$ be the cost of finding the treasure
using Algorithm {\tt Small vision} if $r\leq 1$, using Algorithm {\tt Medium vision} if $1<r<0.9D$, and using Algorithm {\tt Large vision} if $r\geq 0.9D$.
Let $p_0=\lceil \log x \rceil$. Hence $2^{p_0}\leq 2x$. Regardless of the range to which $r$ belongs, the treasure will be found at the latest by the end  of the
execution of the {\bf repeat} loop for $p=p_0$. The total cost of all the executions of the loop for $p\leq p_0$ is at most $12\cdot 2^{p_0} \leq24x$.
Hence the theorem follows from Theorems \ref{small-vision-cost}, \ref{medium-vision-cost} and \ref{large vision}.
\end{proof}

\section{Conclusion}

We designed a treasure hunt algorithm that works for any size of advice and has optimal cost (up to multiplicative constants) whenever $r\leq 1$ or $r\geq 0.9D$, where $r$ is the vision radius and $D$ is an upper bound on the distance between the treasure and the initial position of the agent. For the
intermediate range of vision radius, i.e., when $1<r<0.9D$, our algorithm has almost optimal cost. Finding an algorithm of optimal cost, in this range as well, is a
natural open problem.

As a by-product of our result we can obtain the solution of a natural  related problem of treasure hunt by many agents, without advice. Suppose that $k$ agents, with distinct labels $1,\dots, k$, are collocated at a point $P$ of the plane and have to find a treasure located at an unknown point $Q$ of the plane. $D$ is an upper
bound on the distance from $P$ to $Q$ and $r$ is the vision radius of each agent, but each agent knows only $k$ and its own label, in particular they do not know $D$ and $r$. The treasure is found when some agent gets at distance $r$ from it. The efficiency measure of a treasure hunt algorithm with many agents is the {\em time} of finding the treasure, assuming that agents start simultaneously and walk with constant speed normalized to 1. 

We can use our algorithm to solve this problem as follows. Let $z=\lfloor \log k \rfloor$.
Partition the plane into $2^z$ sectors corresponding to angles of size $2\pi/2^z$, starting from direction North and going counterclockwise.
Call these sectors $S_i$, for $1\leq i \leq 2^z$.
Agent with label $i$, for $1\leq i \leq 2^z$, executes Algorithm {\tt Universal}, supposing that it decoded sector $S_i$ using advice of size $z$. Agents with labels $i$, for $2^z < i \leq k$, remain idle. It follows from our results that the time of finding the treasure is $O(D+\frac{D^2}{kr}(\log D + \log 1/r))$ when $r\leq 1$, it is $O((D+\frac{D^2}{kr})D^{\alpha})$ when $1<r<0.9D$ and the input to the algorithm is $\alpha$, and it is $O(D-r)$ when $r\geq 0.9D$. It also follows that this time is optimal (respectively almost optimal) for treasure hunt by $k$
collocated agents, in the same sense and for the same reasons as understood in this paper.

\end{document}